\documentclass{article}

\usepackage[preprint,nonatbib]{neurips_2026}

\usepackage[utf8]{inputenc}
\usepackage[T1]{fontenc}
\usepackage[hypertexnames=false]{hyperref}
\usepackage{url}
\usepackage{booktabs}
\usepackage{tabularx}
\usepackage{amsfonts}
\usepackage{amsmath}
\usepackage{amssymb}
\usepackage{amsthm}
\usepackage{mathtools}
\usepackage{graphicx}
\usepackage{orcidlink}
\usepackage{microtype}
\usepackage{cleveref}
\usepackage{algorithm}
\usepackage{algorithmic}
\usepackage{enumitem}
\usepackage{tikz}
\usetikzlibrary{fit,backgrounds,positioning,calc}

\newtheorem{theorem}{Theorem}
\newtheorem{lemma}[theorem]{Lemma}
\newtheorem{proposition}[theorem]{Proposition}
\newtheorem{corollary}[theorem]{Corollary}
\newtheorem{definition}[theorem]{Definition}
\newtheorem{assumption}[theorem]{Assumption}
\newtheorem{remark}[theorem]{Remark}

\theoremstyle{definition}
\newcolumntype{Y}{>{\raggedright\arraybackslash}X}

\newcommand{\Rsem}{R_{\mathrm{sem}}}
\newcommand{\Rcrit}{R_{\mathrm{crit}}}
\newcommand{\Dintent}{D_{\mathrm{intent}}}
\newcommand{\dintent}{d_{\mathrm{intent}}}
\newcommand{\dval}{d_{\mathrm{val}}}
\newcommand{\dbeh}{d_{\mathrm{beh}}}
\newcommand{\dQ}{d_Q}
\newcommand{\Csem}{C_{\mathrm{sem}}}
\newcommand{\Ralign}{R_{\mathrm{align}}}
\newcommand{\Wone}{\mathcal{W}_1}
\newcommand{\TV}{\mathrm{TV}}
\newcommand{\E}{\mathbb{E}}
\newcommand{\Prob}{\mathbb{P}}

\newcommand{\calQ}{\mathcal{Q}}

\newcommand{\calO}{\mathcal{O}}
\newcommand{\calA}{\mathcal{A}}
\newcommand{\calS}{\mathcal{S}}
\newcommand{\Qm}{Q_{m,T}}
\newcommand{\Pim}{\Pi_{m,T}}

\newcommand{\hA}{h_A}
\newcommand{\hB}{h_B}

\title{Semantic Rate-Distortion for Bounded Multi-Agent Communication:\\
Capacity-Derived Semantic Spaces and the Communication Cost of Alignment}

\author{
Anthony T Nixon\thanks{
DefSig. Correspondence: \texttt{anthony@defsig.com}. \orcidlink{0000-0002-0289-0954}
}}

\begin{document}
\maketitle

\begin{abstract}
When two agents of different computational capacities interact with the same environment, they need not compress a common semantic alphabet differently; they can induce different semantic alphabets altogether. We show that the quotient POMDP~$\Qm(M)$---the unique coarsest abstraction consistent with an agent's capacity---serves as a \emph{capacity-derived semantic space} for any bounded agent, and that communication between heterogeneous agents exhibits a sharp structural phase transition. Below a critical rate~$\Rcrit$ determined by the quotient mismatch, intent-preserving communication is \emph{structurally impossible}. In the supported one-way memoryless regime, classical side-information coding then yields exponential decay above the induced benchmark. Classical coding theorems tell you the rate once the source alphabet is fixed; \emph{our contribution is to derive that alphabet from bounded interaction itself}.

Concretely, we prove: (1)~a \textbf{fixed-$\varepsilon$ structural phase-transition theorem} whose lower bound is fully general on the common-history quotient comparison used throughout the paper (\S\ref{sec:capacity}); (2)~a \textbf{one-way Wyner-Ziv benchmark identification} on quotient alphabets, with exact converse, exact operational equality for memoryless quotient sources, and an ergodic long-run bridge argued via explicit mixing bounds (\S\ref{sec:wynerziv}); (3)~an \textbf{asymptotic one-way converse} in the shrinking-distortion regime $\varepsilon = O(1/T)$, proved from the message stream and decoder side information rather than from a stronger causal-independence claim (\S\ref{sec:converse}); and (4)~\textbf{alignment traversal bounds} enabling compositional communication through intermediate capacity levels (\S\ref{sec:achievability}). Experiments on eight POMDP environments (including the standard RockSample(4,4) benchmark) illustrate the structural phase transition, a structured-policy benchmark shows that the one-way rate can drop by up to $19\times$ relative to the counting bound, and a dedicated shrinking-distortion sweep matches the regime of the asymptotic converse. Throughout, we make explicit which statements are theorem-level, which are benchmark identifications, and which are qualitative application implications.
\end{abstract}

%% ============================================================
\section{Introduction}
\label{sec:intro}

A human overseer and a frontier AI model; two language models of different sizes; a sensor array and a robotic controller. In each case, agents of different computational capacities must coordinate in a shared environment---and their capacity mismatch creates a fundamental communication challenge. An agent with~$m$ memory nodes perceives the environment through its quotient POMDP~$\Qm(M)$~\cite{nixon2026}: the unique coarsest abstraction consistent with its computational capacity. Two agents with different capacities thus inhabit different \emph{semantic spaces}---even when acting in the same physical world. The quotient is not a design choice; it is mathematically inevitable, determined by which environmental distinctions the agent's memory can sustain.

\paragraph{The central question.} At what rate must agent~$A$ communicate to agent~$B$ to achieve~$\varepsilon$-aligned joint behavior, given their capacity mismatch?

This question arises universally across agent pairs (\Cref{tab:universality}). Classical rate-distortion~\cite{shannon1959} assumes shared codebooks and reconstruction error---assumptions that fail when agents have heterogeneous capacities. The relevant distortion is \emph{operational}: does~$B$ preserve~$A$'s intent? The codebook cannot be universal; it must be \emph{quotient-aware}, designed around the agents' respective equivalence structures. Classical coding theorems characterize the rate once the source alphabet is given. \emph{Our contribution is showing which alphabet each agent's capacity demands}---and that the resulting communication theory produces concrete, experimentally supported predictions of where intent-preserving communication transitions from impossible to achievable.

The paper's backbone is therefore a \emph{fixed-distortion structural theorem}: the capacity gap induces a phase transition even before one asks for asymptotically sharp converses. The one-way Wyner-Ziv reduction then identifies the sharp benchmark once the quotient alphabets are derived, and the shrinking-distortion converse adds an asymptotic sharpening through the actual message stream and decoder side information rather than carrying the full empirical burden of the paper.

\begin{table}[t]
\centering
\footnotesize
\begin{tabular}{@{}lllll@{}}
\toprule
\textbf{Agent pair} & \textbf{Sender ($A$)} & \textbf{Receiver ($B$)} & \textbf{$\Rcrit$ meaning} & \textbf{Instantiation} \\
\midrule
Human $\leftrightarrow$ AI & AI ($m_A$ large) & Human ($m_H$ small) & RLHF feedback floor & \S\ref{sec:applications} \\
Large $\to$ small model & $m_{\text{large}}$ & $m_{\text{small}}$ & Distillation loss floor & Cor.~\ref{thm:align-cost} \\
Sensor $\to$ controller & $m_{\text{sensor}}$ & $m_{\text{ctrl}}$ & Comm.\ bandwidth floor & \S\ref{sec:comm-model} \\
Stepping-stone chain & $m_1 > \cdots > m_k$ & (compositional) & Traversal rate budget & Thm.~\ref{thm:traversal} \\
Equal agents & $m$ & $m$ & $\Rcrit = 0$ (no gap) & Lem.~\ref{lem:refinement} \\
\bottomrule
\end{tabular}
\caption{The framework applies universally to any pair of bounded agents. The critical rate $\Rcrit$ quantifies the structural communication cost of the capacity gap. Below $\Rcrit$, intent-preserving communication is impossible regardless of protocol design.}
\label{tab:universality}
\end{table}

\paragraph{Claim taxonomy.}
The paper makes three kinds of claims. \emph{Theorem-level} claims are proved in the main text or appendix under explicit assumptions. \emph{Benchmark-identification} claims show when the semantic problem can be compared to a classical coding benchmark; these are exact at the converse level under one-way observability, exact operationally for i.i.d.\ quotient sources, and otherwise stated as long-run average bridges. \emph{Application} claims translate the structural results to alignment, routing, or control settings; unless derived directly from a theorem, these are qualitative implications rather than formal guarantees.

\paragraph{Why this is not just Wyner-Ziv on a renamed alphabet.}
Wyner-Ziv theory starts after the source and side-information alphabets are already fixed. In heterogeneous-agent communication, that is precisely the missing object: bounded agents need not share a natural semantic alphabet because their computational capacities induce different quotient partitions. The quotient construction therefore does not merely relabel a classical problem; it identifies \emph{when} a classical side-information theorem becomes available and which distinctions each agent can, or cannot, represent.

\begin{table}[t]
\centering
\footnotesize
\begin{tabularx}{\linewidth}{@{}p{0.16\linewidth}YYY@{}}
\toprule
\textbf{Framework} & \textbf{What is given} & \textbf{What remains open} & \textbf{This paper's role} \\
\midrule
Classical R-D / WZ & Source alphabet and fidelity criterion are specified by the modeler & No account of where heterogeneous agents' semantic alphabets come from & Derives those alphabets from bounded interaction before applying coding theory \\
POMDP abstraction / quotienting & Capacity-dependent state abstractions & No communication theorem between mismatched abstractions & Turns quotient mismatch into a rate theorem with a structural floor \\
IB / task-oriented compression & Relevance variable and task objective are designer-chosen & No receiver-capacity side-information model & Uses receiver quotient classes as capacity-derived side information, not designer-chosen relevance \\
This work & Sender/receiver quotient alphabets are induced by capacity & Communication cost between mismatched semantic spaces & Identifies the structural floor and the one-way benchmark once the alphabets are derived \\
\bottomrule
\end{tabularx}
\caption{Positioning relative to adjacent frameworks. The novelty is not the reuse of classical coding tools per se, but the derivation of agent-dependent semantic alphabets from bounded interaction and the resulting communication problem between them.}
\label{tab:positioning}
\end{table}

\paragraph{Contributions.}
\begin{enumerate}[nosep]
\item \textbf{Capacity-derived semantic spaces and distortion measures} (\Cref{sec:prelim,sec:distortion}): The quotient POMDP as each agent's semantic space, with intent preservation, value-alignment, and quotient morphism distance as operationally meaningful metrics.
\item \textbf{Fixed-$\varepsilon$ structural phase transition} (\Cref{sec:capacity}): Below~$\Rcrit$, semantic distortion is bounded away from zero (\emph{structural impossibility}) under an explicit positive-support condition on merged quotient classes. The constructive exponential upper bound is proved only in the one-way memoryless benchmark regime and is stated as such. Experiments illustrate the random-policy knee relative to the log-cardinality reference without claiming a sharper theorem than the model supports (\Cref{fig:phase-transition}).
\item \textbf{Wyner-Ziv benchmark identification} (\Cref{sec:wynerziv}): In the common-history coarsening setting and under one-way observability, semantic rate-distortion admits a quotient-alphabet WZ benchmark with \emph{exact converse}, \emph{exact i.i.d.\ operational bridge}, and an \emph{ergodic long-run bridge argued with explicit mixing bounds}. Structured visitation lowers the benchmark by up to $19\times$ relative to the counting bound (\Cref{fig:structured-rd}).
\item \textbf{Shrinking-distortion one-way converse} (\Cref{sec:converse}): Asymptotic lower bound from the message stream plus decoder side information in the regime $\varepsilon = O(1/T)$, complemented by a dedicated shrinking-distortion sweep that matches this regime empirically (\Cref{fig:shrinking-eps,rem:fano-regime}).
\item \textbf{Alignment traversal} (\Cref{sec:achievability}): Compositional bounds enable stepping-stone alignment through intermediate capacity levels---bridging the gap incrementally.
\item \textbf{Alignment implications} (\Cref{sec:applications}): Human--AI alignment, model distillation, and sensor-controller communication inherit theorem-backed structural lower bounds, plus constructive upper bounds under explicit memoryless/codebook assumptions.
\end{enumerate}

%% ============================================================
\section{Preliminaries}
\label{sec:prelim}

\begin{definition}[Finite POMDP]
A \emph{finite POMDP} is $M = \langle S, \calA, \calO, P, Z, R, b_0 \rangle$ with finite state, action, and observation sets, transition kernel $P(s'|s,a)$, observation kernel $Z(o|s',a)$, reward function $R$, and initial belief $b_0 \in \Delta(S)$.
\end{definition}

\begin{definition}[Agent Class]
A \emph{stochastic FSC} is $\pi = \langle N, \alpha, \beta, n_0 \rangle$ where $|N| \leq m$, $\alpha(n'|n,o)$ is the node transition, $\beta(a|n)$ the action distribution, and $n_0$ the initial node. $\Pim$ denotes the set of all such FSCs evaluated over horizon~$T$; $\Pi_{m,T,\delta}$ additionally restricts observation resolution to~$\delta$. Histories $h, h' \in \calO^t$ are \emph{bounded-indistinguishable} ($h \equiv_{m,T} h'$) when $\sup_{\pi \in \Pim} \Wone\!\bigl(P_M^\pi(\calO_{t+1:T} \mid h),\, P_M^\pi(\calO_{t+1:T} \mid h')\bigr) = 0$.
\end{definition}

\begin{definition}[Quotient POMDP]
The \emph{quotient POMDP} $\Qm(M)$ has state space $\{[h] : h \in \calO^t\}$ (equivalence classes under $\equiv_{m,T}$) with aggregated transitions induced by~$M$. It is the unique minimal abstraction preserving all observation laws for controllers in~$\Pim$~\cite{nixon2026}; see also~\cite{castro2009} for related POMDP equivalence results. The \emph{probe-exact quotient}~\cite{nixon2026} is the coarsest partition for which $[h]_{\text{probe}} := \{h' : \forall \pi \in \Pim,\; P_M^\pi(\calO_{t+1:T}|h) = P_M^\pi(\calO_{t+1:T}|h')\}$---a self-contained construction within the FSC/POMDP formalism above.
\end{definition}

\begin{remark}[Self-contained well-definedness]
\label{rem:quotient-welldef}
The quotient $\Qm(M)$ has well-defined aggregated transitions because the equivalence $\equiv_{m,T}$ is \emph{right-invariant}: if $h \equiv_{m,T} h'$, then for every observation $o \in \calO$, the extended histories $ho \equiv_{m,T} h'o$. This follows directly from the definition: if all $\pi \in \Pim$ produce identical future observation distributions from $h$ and $h'$, then conditioning on one additional observation preserves this identity (by the chain rule for conditional distributions in the finite POMDP). Right-invariance ensures that the transition $[h] \xrightarrow{o} [ho]$ is independent of the representative, so the quotient POMDP's state transitions are well-defined. The specific quotient facts used later in this paper are collected in \Cref{app:quotient-facts}; the full Myhill--Nerode characterization (uniqueness, minimality, surjectivity of the quotient map) remains in~\cite{nixon2026}.
\end{remark}

\paragraph{Intuition: why different capacities create a communication problem.}
Consider \Cref{fig:quotient-intuition}. A POMDP generates observation histories $h \in \calO^{\leq T}$. Agent~$A$, with $m_A = 16$ memory nodes, can distinguish~781 equivalence classes of histories (its quotient $\calQ_A$); agent~$B$, with $m_B = 1$ node, collapses those same histories into only~289 classes ($\calQ_B$). Because $\calQ_A$ refines $\calQ_B$, every $\calQ_B$-class is the union of one or more $\calQ_A$-classes. The 492 distinctions visible to~$A$ but invisible to~$B$---the $\calQ_A$-subclasses merged within each $\calQ_B$-class---are precisely the communication problem. Unless~$A$ sends enough bits to resolve these merged subclasses, $B$ cannot distinguish histories that~$A$ knows to require different actions: intent-preserving communication is impossible below the rate needed to disambiguate them.

\begin{figure}[t]
\centering
\begin{tikzpicture}[
  >=stealth,
  history/.style={circle, draw, inner sep=1.2pt, fill=black!8, font=\scriptsize},
  classA/.style={rounded corners=3pt, draw=blue!70, thick, inner sep=3pt},
  classB/.style={rounded corners=6pt, draw=red!60, thick, dashed, inner sep=5pt},
]
% Histories
\node[history] (h1) at (0,0) {$h_1$};
\node[history] (h2) at (0.9,0) {$h_2$};
\node[history] (h3) at (2.0,0) {$h_3$};
\node[history] (h4) at (3.5,0) {$h_4$};
\node[history] (h5) at (4.4,0) {$h_5$};
\node[history] (h6) at (5.5,0) {$h_6$};
\node[history] (h7) at (6.4,0) {$h_7$};

% Q_A classes (solid blue)
\begin{scope}[on background layer]
  \node[classA, fit=(h1)(h2), label={[blue!70,font=\tiny]above:$[h]_A^1$}] {};
  \node[classA, fit=(h3), label={[blue!70,font=\tiny]above:$[h]_A^2$}] {};
  \node[classA, fit=(h4)(h5), label={[blue!70,font=\tiny]above:$[h]_A^3$}] {};
  \node[classA, fit=(h6), label={[blue!70,font=\tiny]above:$[h]_A^4$}] {};
  \node[classA, fit=(h7), label={[blue!70,font=\tiny]above:$[h]_A^5$}] {};
\end{scope}

% Q_B classes (dashed red)
\begin{scope}[on background layer]
  \node[classB, fit=(h1)(h2)(h3), label={[red!60,font=\tiny]below:$[h]_B^1$\;\;(3 subclasses $\to$ merged)}] at (0,-0.05) {};
  \node[classB, fit=(h4)(h5)(h6)(h7), label={[red!60,font=\tiny]below:$[h]_B^2$\;\;(3 subclasses $\to$ merged)}] at (0,-0.05) {};
\end{scope}

% Labels
\node[blue!70, font=\small\bfseries] at (8.1, 0.55) {$\calQ_A$};
\node[red!60, font=\small\bfseries] at (8.1, -0.55) {$\calQ_B$};

% Annotation
\draw[<->, thick, black!50] (0,-1.3) -- node[below, font=\scriptsize, text=black!70] {$A$ sees 5 classes; $B$ sees 2. Communication must resolve the merged subclasses.} (6.4,-1.3);
\end{tikzpicture}
\caption{Quotient partitions for two agents of different capacity. Solid blue boxes: $\calQ_A$-classes (finer). Dashed red boxes: $\calQ_B$-classes (coarser). Each $\calQ_B$-class merges multiple $\calQ_A$-classes. Below $\Rcrit = \lceil\log 5 - \log 2\rceil$ bits/step, $B$ cannot distinguish the merged subclasses---some of $A$'s intended actions are irrecoverable.}
\label{fig:quotient-intuition}
\end{figure}
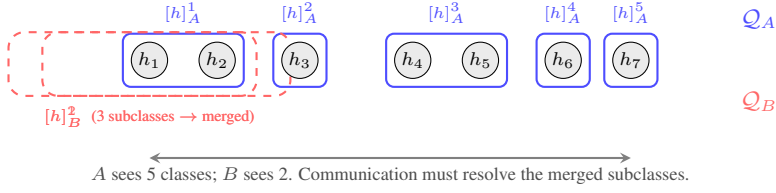

\begin{definition}[Directed Information {\cite{massey1990}}]
\[
I(X^n \to Y^n) := \sum_{t=1}^{n} I(X^t;\, Y_t \mid Y^{t-1}).
\]
\end{definition}

\begin{table}[t]
\centering\footnotesize
\begin{tabular}{@{}ll@{}}
\toprule
\textbf{Symbol} & \textbf{Meaning} \\
\midrule
$M$; $\calS, \calA, \calO$ & POMDP; state, action, observation sets \\
$\Qm(M)$; $\calQ_A, \calQ_B$ & Quotient POMDP; quotient class sets for agents $A$, $B$ \\
$\Pim$ & $m$-node FSCs evaluated over horizon~$T$ \\
$\equiv_{m,T}$ & Bounded indistinguishability (history equivalence under $\Pim$) \\
$\Rcrit$, $\Rsem(\varepsilon)$ & Critical rate, semantic rate-distortion function \\
$\dintent$, $\Dintent$ & Per-history and expected intent distortion \\
$\dval$, $\dbeh$, $\dQ$ & Value-alignment, behavioral, quotient morphism distances \\
$\hA$, $\hB$ & Entropy rates of quotient processes $\{Q_t^A\}$, $\{Q_t^B\}$ \\
$\bar{h}(Q_A \mid Q_B)$ & Conditional entropy rate (WZ benchmark) \\
$c_M$ & Minimum inter-class intent distortion gap \\
$I(X^n \to Y^n)$ & Directed information \\
\bottomrule
\end{tabular}
\caption{Core notation. See \Cref{sec:prelim,sec:distortion} for formal definitions.}
\label{tab:notation}
\end{table}

%% ============================================================
\section{Semantic Distortion Measures}
\label{sec:distortion}

Classical rate-distortion uses reconstruction error. For multi-agent communication, we need distortion capturing \emph{meaning preservation}.

\begin{definition}[Value-Alignment Distortion]
$\dval(\pi_A, \pi_B; M) := \sup_{R:\, L_R \leq 1} |V_M^{\pi_A}(R) - V_M^{\pi_B}(R)|$, a pseudometric on policies.
\end{definition}

\begin{definition}[Behavioral Divergence]
$\dbeh(\pi_A, \pi_B; M) := \Wone\!\bigl(P_M^{\pi_A}(\calO^T),\, P_M^{\pi_B}(\calO^T)\bigr)$.
\end{definition}

\begin{proposition}[Value Bound from Behavior]
\label{prop:val-bound}
$\dval(\pi_A, \pi_B) \leq L_R \cdot T \cdot \dbeh(\pi_A, \pi_B)$.
\end{proposition}
\begin{proof}
By $L_R$-Lipschitz reward (Assumption~\ref{ass:lip}): $|\bar{R}_M(h_t, \pi_A) - \bar{R}_M(h_t, \pi_B)| \leq L_R \cdot \Wone(P_M^{\pi_A}(\calO_{t+1:T}|h_t), P_M^{\pi_B}(\calO_{t+1:T}|h_t))$ per step. Sum over $T$ steps (see also \Cref{lem:lip-prop}). This is self-contained; see~\cite{nixon2026} for the general FSC setting.
\end{proof}

\begin{definition}[Intent and Intent Distortion]
Agent~$A$'s \emph{intent} at history~$h$ is $\mathrm{Intent}_A(h) := \bigl(b_h^A,\; \pi_A(\cdot \mid h),\; V^{\pi_A}(h)\bigr)$. The \emph{intent preservation distortion} is:
\[
\dintent(A, B \mid h) := \Wone(b_h^A, b_h^B) + \TV\bigl(\pi_A(\cdot|h), \pi_B(\cdot|h)\bigr) + |V^{\pi_A}(h) - V^{\pi_B}(h)|,
\]
with expected distortion $\Dintent(A, B) := \E_{h \sim P_M^{\pi_A}}[\dintent(A, B \mid h)]$.
\end{definition}

\begin{remark}[Component Scales and Reward Normalization]
\label{rem:scales}
The three components of~$\dintent$ have different scales: $\Wone$ on simplices over $n$ states takes values in $[0, 2]$; TV distance in $[0, 1]$; the value difference depends on the reward scale. We assume throughout that $\|R\|_\infty \leq 1$ (i.e., rewards are normalized to $[0,1]$), ensuring the value-difference component $|V^{\pi_A}(h) - V^{\pi_B}(h)| \leq T$ is bounded; without this normalization, $\dintent$ would be unbounded. Under this assumption, $\dintent$ is a well-defined bounded distortion measure. The critical rate~$\Rcrit$ depends only on quotient structure (class counts), not on distortion scale. Our experiments use a simplified two-term proxy; see \Cref{app:experiments}.
\end{remark}

\begin{proposition}[Proxy Bound (deterministic policies)]
\label{prop:proxy}
Let $\dintent^{\mathrm{exp}}(h) := \Wone(b_h^A, b_h^B) + \lambda \cdot \mathbf{1}\{a_A \neq a_B\}$ be the two-term experimental proxy with $\lambda > 0$. For deterministic policies (where action mismatch implies $\TV(\pi_A(\cdot|h), \pi_B(\cdot|h)) = 1 \geq \lambda$):
\[
\dintent^{\mathrm{exp}}(h) \;\leq\; \dintent(h) \;\leq\; \dintent^{\mathrm{exp}}(h) + |V^{\pi_A}(h) - V^{\pi_B}(h)|.
\]
In particular, the critical rate~$\Rcrit$ (which depends on quotient class counts, not distortion scale; see \Cref{cor:exact-rcrit}) is identical under both measures. The phase transition location and exponential decay exponent are invariant to the choice of distortion metric; only the pre-exponential constant~$c_M$ changes.
\end{proposition}

\begin{proof}
The lower bound holds when $\lambda \leq \TV(\pi_A(\cdot|h), \pi_B(\cdot|h))$ whenever $a_A \neq a_B$; this is satisfied for deterministic policies (where action mismatch implies $\TV = 1$) and in our experiments with $\lambda = 0.5$. The upper bound follows by the triangle inequality. Since $\Rcrit = \bar{h}(Q_A \mid Q_B)$ depends only on the quotient process entropies and not on the distortion function, the phase transition location is metric-invariant. In the supported one-way memoryless regime, the constructive error exponent depends on surplus rate above the relevant benchmark, not on~$c_M$.
\end{proof}

\begin{definition}[Quotient Morphism Distance]
A \emph{quotient morphism} $\varphi: \calQ_A \to \calQ_B$ preserves the initial class and respects transitions. The morphism distance is:
\[
\dQ(A, B \mid M) := \inf_{\varphi:\, \calQ_A \to \calQ_B}\; \sup_{[h] \in \calQ_A} \Wone\!\bigl(P_M(\cdot|[h]),\, P_M(\cdot|\varphi([h]))\bigr).
\]
When $\calQ_A$ refines $\calQ_B$, the inclusion map is a morphism and $\dQ = 0$.
\end{definition}

%% ============================================================
\section{Multi-Agent Communication Model}
\label{sec:comm-model}

\paragraph{Setup.}
Environment~$M$; agent~$A$ with capacity $(m_A, T_A)$; agent~$B$ with capacity $(m_B, T_B)$. At each step~$t$, $A$ observes $O_t^A$, takes action $a_t^A$, and may send message $M_t \in \{1, \ldots, 2^R\}$ to~$B$ over a noiseless channel at rate~$R$ bits/step. $B$ observes $O_t^B$, receives $M_t$, and takes action $a_t^B$.

\begin{definition}[Communication Protocol]
\label{def:protocol}
A \emph{protocol} $(E, D)$ consists of:
\begin{itemize}[nosep]
\item An \emph{encoder} $E_t: (\calO^A)^{\leq t} \to \{1,\ldots,2^R\}$ with $H(M_t \mid M^{t-1}) \leq R$.
\item A \emph{bounded decoder} $D = (\pi_B, U)$, where $\pi_B \in \Pi_{m_B,T_B}$ is an $m_B$-node FSC and $U: \{1,\ldots,2^R\} \to \Theta_{m_B}$ updates the FSC's parameters upon each received message. Here $\Theta_{m_B} := \Delta(\calA)^{m_B} \times \Delta(\{1,\ldots,m_B\})^{m_B \times |\calO|}$ is the space of action distributions and node-transition matrices for an $m_B$-node FSC. The FSC~$\pi_B$ takes inputs $(n_t^B, O_t^B)$ at each step~$t$, where $n_t^B \in \{1,\ldots,m_B\}$ is the current node and $O_t^B \in \calO^B$ is~$B$'s current observation; this is implicit in the FSC definition of \Cref{sec:prelim}. Messages may \emph{reconfigure} $\pi_B$'s parameters but cannot expand its $m_B$-node memory capacity.
\end{itemize}
\end{definition}

\begin{assumption}[Common-history coarsening regime]
\label{ass:common-history}
The theorem-level comparisons between $\calQ_A$ and $\calQ_B$ are stated for a single underlying history process $h_t \in \calO^t$. The sender and receiver quotient processes are $Q_t^A := [h_t]_{m_A,T_A}$ and $Q_t^B := [h_t]_{m_B,T_B}$ on this same history space, so when $m_A \geq m_B$ the canonical map of \Cref{prop:canonical-coarsening} gives $Q_t^B = \kappa_{A \to B}(Q_t^A)$ pointwise. The more general notation $(O_t^A, O_t^B)$ above is retained to distinguish encoder and decoder roles; distinct-sensor models are covered by the theorem-level results only when they induce this same coarsening relation.
\end{assumption}

\begin{assumption}[Intent measurability on quotient classes]
\label{ass:intent-meas}
Under the common-history regime of \Cref{ass:common-history}, the sender intent $\mathrm{Intent}_A(h)$ is constant within each $\calQ_A$-class: if $[h]_{m_A,T_A} = [h']_{m_A,T_A}$, then $\mathrm{Intent}_A(h) = \mathrm{Intent}_A(h')$. Equivalently, $\dintent$ induces a well-defined distortion on $\calQ_A \times \calQ_A$.
\end{assumption}

\begin{remark}[Scope of intent measurability]
\label{rem:intent-meas}
\Cref{ass:intent-meas} is needed for the WZ benchmark identification (\Cref{sec:wynerziv}), the shrinking-distortion converse (\Cref{sec:converse}), and the Blahut-Arimoto benchmark computation, all of which use~$\dintent$ as a class-level distortion matrix. It is \emph{not} needed for the structural lower bound of \Cref{thm:capacity}(i), whose pigeonhole argument uses the history-level separation~$c_M$. A sufficient condition is that the sender's policy~$\pi_A$ is \emph{quotient-compatible}: $\pi_A(\cdot \mid h) = \pi_A(\cdot \mid h')$ whenever $[h]_{m_A,T_A} = [h']_{m_A,T_A}$. This holds when~$\pi_A$ is derived from the quotient POMDP~$\Qm(M)$ or is a belief-based policy, since distinct quotient classes have distinct beliefs (\Cref{rem:cM-positive}).
\end{remark}

\begin{definition}[Semantic Rate-Distortion Function]
\[
\Rsem(\varepsilon) := \inf\bigl\{R : \exists\, \text{protocol } (E,D) \text{ with } \Dintent(A, B \mid E, D) \leq \varepsilon\bigr\}.
\]
\end{definition}

%% ============================================================
\section{Semantic Channel Capacity}
\label{sec:capacity}

\begin{theorem}[Semantic Channel Capacity]
\label{thm:capacity}
For agents~$A$, $B$ communicating over a noiseless channel at rate~$R$, define $\Csem(R) := \inf\{\Dintent \text{ achievable at rate } R\}$. Then:

\begin{enumerate}[label=\textup{(\roman*)},nosep]
\item \textbf{Threshold behavior.}\;
$\Csem(R) \geq c_M > 0$ \;for\; $R < \Rcrit$ whenever there exists a reachable $\calQ_B$-class of positive stationary mass that contains more than $2^R$ reachable $\calQ_A$-subclasses. Under uniform quotient visitation this yields the structural threshold $\Rcrit := \log|\calQ_A| - \log|\calQ_B|$. Here $c_M := \min_{\substack{[h]_A \neq [h']_A \\[2pt] [h]_A,\,[h']_A \subseteq [h]_B}} \;\min_{\substack{h \in [h]_A,\; h' \in [h']_A}} \dintent\!\bigl(\mathrm{Intent}_A(h),\, \mathrm{Intent}_A(h')\bigr)$ is the minimum intent distortion between any pair of histories in distinct merged subclasses.

\item \textbf{Exponential decay in the one-way memoryless regime.}\;
Under Assumptions~\ref{ass:common-history}, \ref{ass:intent-meas}, \ref{ass:one-way}, and~\ref{ass:lip}, if the joint quotient source $(Q_t^A,Q_t^B)$ is i.i.d., then every rate
\[
R \;>\; R_{\mathrm{WZ}}(0) \;=\; H(Q_A \mid Q_B)
\]
admits a block semantic protocol whose distortion decays exponentially in~$T$:
\[
\Csem(R) \;\leq\; L_R \cdot T \cdot 2^{-T E_{\mathrm{WZ}}(R)},
\]
for some classical WZ reliability exponent $E_{\mathrm{WZ}}(R) > 0$. In the uniform-cardinality special case, $H(Q_A \mid Q_B) = \log|\calQ_A| - \log|\calQ_B| = \Rcrit$.

\item \textbf{Perfect alignment.}\;
$\Csem(R) = 0$ \;if\; $R \geq \log|\calQ_A|$.
\end{enumerate}
\end{theorem}

\begin{proof}[Proof sketch]
\textbf{(i)} Pigeonhole on a positively supported merged $\calQ_B$-class: if that class contains more than $2^R$ reachable $\calQ_A$-subclasses, two of them must share the same message/side-information pair, forcing positive-probability confusion and therefore $\dintent \geq c_M$. Full proof in \Cref{app:proof-51i}.
\textbf{(ii)} In the one-way i.i.d.\ regime, \Cref{thm:wz,cor:iid-achievability} reduce the problem to lossless WZ coding on quotient alphabets; classical reliability exponents then give exponentially decaying block error above $H(Q_A \mid Q_B)$, and \Cref{lem:lip-prop} propagates that error to semantic distortion. Full proof in \Cref{app:proof-51ii}.
\textbf{(iii)} $R \geq \log|\calQ_A|$ specifies $A$'s class exactly.
\end{proof}

The phase transition is structural: when $\calQ_A$ strictly refines $\calQ_B$, distinctions visible to~$A$ are invisible to~$B$ regardless of protocol---a direct consequence of the quotient refinement $\calQ_A \preceq \calQ_B$ in the common-history comparison regime of \Cref{ass:common-history}. Extended discussion of $c_M$ positivity, the role of $B$'s bounded memory, the positive-support requirement in part~(i), and the source-distribution dependence of the constructive benchmark appear in \Cref{app:capacity-remarks}.

%% ============================================================
\section{Shrinking-Distortion Converse}
\label{sec:converse}

This section is an \emph{asymptotic sharpening} of the fixed-$\varepsilon$ structural theorem above. The main operational statement of the paper remains \Cref{thm:capacity}; the theorem here gives a shrinking-distortion lower bound from the actual message stream and decoder side information, avoiding a stronger causal-independence claim than the model justifies.

\begin{assumption}[Quotient Process Regularity]
\label{ass:quotient-reg}
Under Assumption~\ref{ass:common-history} and the communication protocol, the joint quotient process $(Q_t^A, Q_t^B)$ induced by the common history source converges to a stationary distribution and is ergodic. This holds when the underlying POMDP has a unique stationary distribution over states and the encoder is time-invariant.
\end{assumption}

\begin{theorem}[Asymptotic One-Way Converse (Shrinking-Distortion Regime)]
\label{thm:converse}
Under Assumptions~\ref{ass:common-history}, \ref{ass:intent-meas}, \ref{ass:quotient-reg}, and~\ref{ass:one-way}, consider a sequence of horizons~$T$ and protocols $(E^{(T)}, D^{(T)})$ with per-step rates~$R_T$ and expected intent distortions~$D_T \leq \varepsilon_T$. Let $\varepsilon_T' := 2\varepsilon_T / c_M$ and assume $T\varepsilon_T' \leq 1/2$. Then
\[
R_T \;\geq\; (h_A - h_B) - h(\varepsilon_T') - \varepsilon_T' \log|\calQ_A| - o_T(1),
\]
where $h_A, h_B$ are the entropy rates of~$\{Q_t^A\}$, $\{Q_t^B\}$ (Remark~\ref{rem:entropy-rate}) and $h(\cdot)$ is the binary entropy function. In particular, if $\varepsilon_T' \to 0$, then
\[
\liminf_{T \to \infty} R_T \;\geq\; h_A - h_B.
\]
Under near-uniform quotient distributions ($h_A \approx \log|\calQ_A|$, $h_B \approx \log|\calQ_B|$), the right-hand side reduces to the log-cardinality reference $\log|\calQ_A| - \log|\calQ_B|$.
\end{theorem}

\begin{proof}[Proof sketch]
$H(Q_A^T \mid Q_B^T) = T(h_A - h_B) + o(T)$ by quotient coarsening on the common history source. Since the decoder sees $(M^T, Q_B^T)$, the rate constraint gives $TR_T \geq I(Q_A^T; M^T \mid Q_B^T)$. A semantic Fano argument based on nearest-intent decoding from $(M^T, Q_B^T)$ turns the shrinking distortion assumption $\varepsilon_T = O(1/T)$ into the entropy penalty $h(\varepsilon_T') + \varepsilon_T' \log|\calQ_A| + o_T(1)$, producing the stated lower bound. Full proof in \Cref{app:proof-61}.
\end{proof}

\begin{remark}[Fano regime restriction]
\label{rem:fano-regime}
The Semantic Fano inequality (\Cref{lem:fano}) underlying \Cref{thm:converse} requires $\varepsilon = O(1/T)$ for the block error bound to be non-vacuous. This means the converse is formally operative only when the distortion tolerance shrinks with horizon length. For fixed~$\varepsilon$ and growing~$T$, the bound degrades. We therefore separate the empirical roles of the figures: \Cref{fig:phase-transition} remains a fixed-$\varepsilon$ illustration of the broader structural phase transition, while \Cref{fig:shrinking-eps} provides one empirical illustration in the same shrinking-distortion regime as the converse. Relaxing the converse to constant-$\varepsilon$ regimes is an open problem; a possible route is via the blowing-up lemma~\cite{csiszarkorner2011}.
\end{remark}

\begin{remark}[Phase transition is regime-independent]
\label{rem:constant-eps}
The $\varepsilon = O(1/T)$ restriction applies \emph{only} to the shrinking-distortion converse (\Cref{thm:converse}). The structural lower bound in \Cref{thm:capacity}(i) is a fixed-$\varepsilon$ statement, while the exponential achievability theorem of \Cref{thm:capacity}(ii) belongs only to the one-way memoryless benchmark regime. The experiments therefore illustrate the structural phase transition at fixed $\varepsilon = 0.1$ without being read as a pointwise empirical proof of the shrinking-distortion converse.
\end{remark}

%% ============================================================
\section{Wyner-Ziv Reduction}
\label{sec:wynerziv}

The semantic rate-distortion problem admits a one-way \emph{benchmark identification} with the classical Wyner-Ziv problem~\cite{wynerziv1976} (source coding with decoder side information) on quotient alphabets. The exactness split is as follows: the semantic converse matches the WZ converse exactly; the operational bridge is exact for memoryless (i.i.d.) quotient sources; \textbf{for general ergodic sources, the achievability direction is argued (not proved) via a causal pipeline with explicit mixing-time bounds} (\Cref{app:proof-wz})---a fully explicit FSC-level pathwise equivalence remains open. \emph{This identification is stated only in the common-history coarsening regime of \Cref{ass:common-history} and under one-way observability} (\Cref{ass:one-way}); the two-way and genuinely distinct-sensor cases remain open (\Cref{app:two-way}).

\begin{assumption}[One-Way Observability]
\label{ass:one-way}
Agent~$B$'s actions do not causally affect the sender-side source history:
\[
P(O_t^A \mid s_t, a_t^A, a_t^B) = P(O_t^A \mid s_t, a_t^A).
\]
\end{assumption}

This holds whenever~$A$ is a sensor or instructor and~$B$ is an actuator or learner---the natural model for human--AI alignment, where the AI acts and the human provides feedback that does not directly alter the source observed by the sender. Together with \Cref{ass:common-history}, it gives the coarsening-side-information setting used in the WZ and shrinking-distortion converse sections. It excludes fully cooperative Dec-POMDPs where both agents jointly affect the shared state and neither history is a deterministic coarsening of the other.

\begin{remark}[When distinct-sensor models reduce to common-history coarsening]
\label{rem:distinct-sensor}
The common-history regime (\Cref{ass:common-history}) is not as restrictive as it may first appear. A genuinely distinct-sensor model $(O_t^A, O_t^B)$ reduces to the common-history setting whenever there exists a shared sufficient statistic---for example, when both agents observe noisy versions of a common underlying signal~$X_t$ and one agent's observation is a stochastic degradation of the other's. Concretely, if $O_t^B - O_t^A - X_t$ forms a Markov chain (i.e., $B$'s observation is a further corruption of~$A$'s), then the quotient induced by~$B$'s observation history coarsens that of~$A$'s, recovering the coarsening relation of \Cref{prop:canonical-coarsening}. This covers sensor-controller pairs with shared physics but different sensor quality, and teacher-student pairs where the student sees a lossy version of the teacher's input. The genuinely non-reducible case---where $A$ and $B$ observe complementary aspects of the environment with no dominance relation---remains open and likely requires the two-way directed-information framework of \Cref{app:two-way}.
\end{remark}

\begin{proposition}[One-Way Wyner-Ziv Benchmark Identification]
\label{thm:wz}
Under Assumptions~\ref{ass:common-history}, \ref{ass:intent-meas}, \ref{ass:quotient-reg}, and~\ref{ass:one-way}:
\begin{enumerate}[nosep,label=(\roman*)]
\item \textbf{(Proved.)} Any semantic protocol achieving distortion~$D$ induces a valid WZ code: $\Rsem(D) \geq R_{\mathrm{WZ}}(D;\; \calQ_A, \calQ_B, \dintent)$.
\item \textbf{(Proved, i.i.d.; argued, ergodic.)} Any WZ code on $(\calQ_A, \calQ_B, \dintent)$ can be implemented as a causal pipelined semantic protocol. For memoryless (i.i.d.) quotient sources, the pipeline achieves exact WZ distortion (\Cref{cor:iid-achievability}). For general ergodic sources, the \emph{long-run time-averaged} semantic distortion converges to $D_{\mathrm{WZ}}$ via the ergodic theorem with explicit mixing-time bounds (\Cref{app:proof-wz}); a fully explicit FSC-level equality remains open.
\end{enumerate}
\end{proposition}

\begin{corollary}[Exact I.I.D.\ Operational Achievability]
\label{cor:iid-achievability}
Under the conditions of \Cref{thm:wz}, if the quotient process $\{Q_t^A\}$ is i.i.d.\ (memoryless source), then the causal pipeline achieves $\Rsem(D) = R_{\mathrm{WZ}}(D)$ exactly. Successive blocks are independent, so $\theta^{(k+1)} = f(\hat{Q}_A^{(k-1)}, Q_B^{(k-1)})$ is independent of $Q_A^{(k+1)}$, and each block's distortion equals $D_{\mathrm{WZ}}$ exactly.
\end{corollary}

\begin{remark}[Ergodic long-run achievability (argued)]
\label{rem:wz-open}
The converse leg of \Cref{thm:wz} is exact. The only non-exact leg is the ergodic achievability direction: the causal pipeline introduces a two-block delay, so $\theta^{(k+1)}$ is correlated with $Q_A^{(k+1)}$ through the source dependence. Under geometric mixing with rate $\rho < 1$ (guaranteed by \Cref{ass:quotient-reg} on a finite state space), the correlation decays as $\rho^{2n}$ across the two-block gap. For any $\delta > 0$, choosing blocklength $n \geq (2\log(1/\delta))/\log(1/\rho)$ ensures per-block distortion deviation $\leq \delta$, yielding $\bar{D}_K \leq D_{\mathrm{WZ}} + \delta + D_0/K$. The ergodic theorem then gives almost-sure convergence of the time-averaged distortion (\Cref{app:proof-wz}). A fully explicit pathwise FSC-level equivalence remains open.
\end{remark}

\begin{corollary}[One-Way Critical-Rate Benchmark]
\label{cor:exact-rcrit}
Under Assumptions~\ref{ass:common-history}, \ref{ass:intent-meas}, \ref{ass:quotient-reg}, and~\ref{ass:one-way} (i.e., the conditions of \Cref{thm:wz}), denote the conditional entropy rate $\bar{h}(Q_A \mid Q_B) := \lim_{T \to \infty} \frac{1}{T} H(Q_A^T \mid Q_B^T)$. Then the lossless Wyner-Ziv benchmark on quotient alphabets is
\[
\Rcrit^{\mathrm{WZ}} \;:=\; R_{\mathrm{WZ}}(0) \;=\; \bar{h}(Q_A \mid Q_B) \;=\; \hA - \hB,
\]
where the last equality uses $H(Q_A^T, Q_B^T) = H(Q_A^T)$ (since $Q_B^T$ is determined by $Q_A^T$ via the coarsening $\calQ_A \preceq \calQ_B$). For i.i.d.\ quotient distributions, $\bar{h}(Q_A \mid Q_B) = H(Q_A \mid Q_B) = H(Q_A) - H(Q_B)$. The log-cardinality form $\log|\calQ_A| - \log|\calQ_B|$ is the further special case of uniform distributions.
Under the i.i.d.\ exact bridge and the ergodic long-run bridge of \Cref{thm:wz}, $\Rcrit^{\mathrm{WZ}}$ is the one-way benchmark for the semantic critical rate.
\end{corollary}

\begin{table}[t]
\centering\footnotesize
\begin{tabular}{@{}lll@{}}
\toprule
\textbf{Form of $\Rcrit$} & \textbf{Expression} & \textbf{Regime} \\
\midrule
Log-cardinality (worst-case) & $\log|\calQ_A| - \log|\calQ_B|$ & Uniform quotient dist.\ \\
WZ lossless rate (i.i.d.) & $H(Q_A \mid Q_B)$ & i.i.d.\ source \\
Conditional entropy rate & $\bar{h}(Q_A \mid Q_B) = \hA - \hB$ & Stationary ergodic \\
\bottomrule
\end{tabular}
\caption{Hierarchy of critical-rate characterizations, from conservative (top) to sharp (bottom). The log-cardinality form is policy-independent; the entropy-rate form depends on the joint quotient dynamics. Numerical values for Chain5 appear in \Cref{app:experiments}.}
\label{tab:rcrit}
\end{table}

\paragraph{Source-distribution dependence.}
The benchmark $\Rcrit^{\mathrm{WZ}} = \bar{h}(Q_A \mid Q_B)$ depends on the policy-induced visitation over quotient classes; the log-cardinality form is the worst-case (uniform) upper bound. \Cref{tab:rcrit} summarizes the hierarchy. Experiments use random policies (near-uniform visitation), so the empirical knee tracks the log-cardinality bound; under structured policies the true benchmark can be substantially lower (see \Cref{app:experiments}).

The WZ single-letter formula, gap closure, and inherited strong converses/error exponents/finite-blocklength bounds follow from the reduction; see \Cref{app:wz-corollaries} for details.

\begin{proposition}[Encoder-Only Penalty Relative to the WZ Benchmark]
\label{prop:ib-subsumption}
Let $X = Q_A$ and $Y = Q_B$ with the deterministic coarsening $f\colon \calQ_A \to \calQ_B$, so that $Q_B - Q_A - T$ is a Markov chain for any bottleneck variable~$T$.
For any encoder $T$ inducing distortion $D(T)$ under~$\dintent$,
\begin{equation}\label{eq:ib-wz}
I(Q_A;\, T)
\;\geq\;
R_{\mathrm{WZ}}\!\bigl(D(T)\bigr) \;+\; I(Q_B;\, T),
\end{equation}
where $R_{\mathrm{WZ}}(D) = \min_{p(u \mid q_A):\, \E[\dintent] \leq D}\bigl[I(Q_A; U) - I(Q_B; U)\bigr]$ is the one-way WZ benchmark on quotient alphabets (\Cref{cor:gap-closure}).
Consequently, if $T_\Delta$ is IB-optimal at relevance~$\Delta$ for
\[
R_{\mathrm{IB}}(\Delta) := \min_{p(t \mid q_A):\; I(T;\,Q_B) \geq \Delta}\, I(Q_A;\, T),
\]
and $D^{*}(\Delta)$ is the distortion induced by~$T_\Delta$, then
\[
R_{\mathrm{IB}}(\Delta)
\;\geq\;
R_{\mathrm{WZ}}\!\bigl(D^{*}(\Delta)\bigr) \;+\; \Delta.
\]
At the lossless endpoint, the encoder-only / side-information gap is exact:
\[
R_{\mathrm{enc}}(0) - R_{\mathrm{WZ}}(0)
\;=\;
H(Q_A) - H(Q_A \mid Q_B)
\;=\;
H(Q_B),
\]
where $R_{\mathrm{enc}}(0) = H(Q_A)$ is the standard lossless rate without decoder side information.
\end{proposition}
\begin{proof}[Proof sketch]
For any encoder~$T$, the induced pair $(T, D(T))$ is feasible for the WZ objective at distortion~$D(T)$, so by definition of~$R_{\mathrm{WZ}}$,
\[
R_{\mathrm{WZ}}\!\bigl(D(T)\bigr) \;\leq\; I(Q_A;\, T) - I(Q_B;\, T).
\]
Rearranging gives~\eqref{eq:ib-wz}. Applying the inequality to an IB-optimal encoder~$T_\Delta$ yields the second claim. For the lossless endpoint, the encoder-only rate is the standard lossless source-coding rate $R_{\mathrm{enc}}(0) = H(Q_A)$, while \Cref{cor:exact-rcrit} gives $R_{\mathrm{WZ}}(0) = H(Q_A \mid Q_B)$; subtracting yields the exact gap~$H(Q_B)$ because $Q_B$ is a deterministic function of~$Q_A$.
\end{proof}

\begin{remark}[IB vs.\ semantic R-D: what is derived vs.\ pre-specified]
\label{rem:ib-distinction}
In the IB framework, the relevance variable~$Y$ is chosen by the designer; in our framework, $Y = Q_B$ is \emph{derived} from agent~$B$'s computational capacity via the quotient functor~$Q$.
\Cref{prop:ib-subsumption} shows that once this identification is made, encoder-only bottlenecks such as IB-style compressions pay an additive penalty relative to the side-information-aware WZ benchmark; at zero distortion, that penalty is exactly the side-information entropy~$H(Q_B)$.
\end{remark}

%% ============================================================
\section{Achievability and Constructive Schemes}
\label{sec:achievability}

\begin{theorem}[Alignment Traversal]
\label{thm:traversal}
For agent classes $\Pi_A$ and $\Pi_B$ connected by intermediate classes $\{\Pi_i\}$ with quotient morphisms $\varphi_i: \calQ_i \to \calQ_{i+1}$ having Lipschitz constants~$L_i$,
\[
\dQ(A, B \mid M) \leq \sum_{i=1}^{k-1} L_i \cdot \dQ(\Pi_i, \Pi_{i+1} \mid M),
\]
and therefore
\[
\Rsem(A \to B;\, \varepsilon) \leq \sum_i \Rsem(\Pi_i \to \Pi_{i+1};\, \varepsilon/k).
\]
Proof in \Cref{app:proof-traversal}.
\end{theorem}

\paragraph{Single-letter achievability.}
The worst-case over $\Pim$ prevents single-letter coding theorems. Under a memoryless reference policy distribution $\mu \in \Delta(\Pim)$ (\textbf{strong assumption}; the one-way WZ reduction removes the memoryless restriction on the rate-distortion benchmark under one-way observability, but the single-letter \emph{formula} below is exact only for memoryless~$\mu$), the average-case rate-distortion decomposes into independent per-step problems:

\begin{theorem}[Memoryless Single-Letter Achievability]
\label{thm:single-letter}
If the reference policy distribution~$\mu$ is memoryless, then
\[
R_\mu(\varepsilon) = \min_{P(M \mid \mathrm{Intent}):\, \E[\dintent] \leq \varepsilon} I(\mathrm{Intent};\, M).
\]
Under one-way observability, the WZ benchmark of \Cref{thm:wz} governs the general stationary-ergodic problem at the benchmark level, but the single-letter formula above is exact only for memoryless~$\mu$.
\end{theorem}

\begin{proof}[Proof sketch]
Under memoryless~$\mu$, the intent sequence $\{\mathrm{Intent}_t\}$ is i.i.d.\ (each drawn independently from~$\mu$'s induced distribution on quotient classes). The $T$-step mutual information decomposes: $I(\mathrm{Intent}^T; M^T) = \sum_t I(\mathrm{Intent}_t; M_t)$, and standard single-letter rate-distortion theory~\cite{cover2006} applies to each term.
\end{proof}

\begin{remark}[Beyond the memoryless reference distribution]
For Markovian or more general reference processes, the i.i.d.\ single-letter objective can still be used as a conservative constructive upper bound by restricting attention to memoryless encoders that ignore temporal correlation. The true operational rate can only be lower; the theorem above does not claim an exact single-letter characterization outside the memoryless setting.
\end{remark}

\begin{definition}[Semantic Codebook Construction]
\label{def:codebook}
Compute quotient $\calQ_A$; for each class $[h]$, compute centroid intent $\bar{I}([h])$; quantize to $2^R$ codewords via $k$-means++~\cite{arthur2007}. Encoder maps $[h]$ to nearest codeword; decoder maps codeword to recommended action distribution. This achieves $\Dintent \leq O(|\calQ_A|^{1/d} \cdot 2^{-R/d})$, where $d$ is the effective dimension of the intent space, i.e.\ the number of independent coordinates needed to represent the intent vectors up to negligible residual variance (\Cref{app:codebook-proof}).
\end{definition}

%% ============================================================
\section{Applications to Alignment}
\label{sec:applications}

Model the human as agent~$H$ with capacity $(m_H, T_H)$ and the AI as agent~$A$ with $(m_A, T_A) \gg (m_H, T_H)$. This section derives \emph{the form of the alignment cost} under the quotient model. The structural lower bound (\Cref{thm:capacity}(i)) and the constructive upper bound (\Cref{thm:single-letter}) are theorem-level given their assumptions. The practical implications for RLHF, debate, and routing are \emph{qualitative}: they predict scaling forms (linear in capacity gap, logarithmic in accuracy) but do not yield computable numerical bounds until the effective quotient cardinalities $|\calQ_A|$, $|\calQ_H|$ can be estimated for real systems---an open problem discussed in \Cref{sec:discussion}.

\begin{corollary}[Alignment Cost Scaling]
\label{thm:align-cost}
The alignment rate $\Ralign(\varepsilon) := \Rsem(A \to H;\, \varepsilon)$ satisfies:
\begin{enumerate}[nosep,label=(\roman*)]
\item \textbf{Structural lower bound (from \Cref{thm:capacity}(i)):} $\Ralign(\varepsilon) \geq \Rcrit = \log|\calQ_A| - \log|\calQ_H|$ for any $\varepsilon < c_M$.
\item \textbf{Constructive upper bound (memoryless/codebook regime):} Under the memoryless reference-distribution assumption of \Cref{thm:single-letter} and the codebook construction of \Cref{def:codebook}, there exists a quotient-aware protocol with
\[
\Ralign^{\mathrm{cb}}(\varepsilon) \;\leq\; d_{\mathrm{eff}} \cdot \log(|\calQ_A|^{1/d_{\mathrm{eff}}} / \varepsilon) + O(1),
\]
giving an $O(\log(1/\varepsilon))$ accuracy cost beyond the structural floor.
\end{enumerate}
Together: the theorem-backed part is the structural floor $\Ralign(\varepsilon) \geq \log|\calQ_A| - \log|\calQ_H|$, while the codebook scheme provides a constructive rate in a restricted regime.
\end{corollary}
\begin{proof}
Part~(i) is \Cref{thm:capacity}(i) applied to the $(A, H)$ pair. Part~(ii) inverts the codebook bound $\Dintent \leq C_d \cdot |\calQ_A|^{1/d} \cdot 2^{-R/d}$ (\Cref{app:codebook-proof}): setting $\varepsilon = C_d \cdot |\calQ_A|^{1/d} \cdot 2^{-R/d}$ and solving for~$R$ gives $R = d \cdot \log(C_d \cdot |\calQ_A|^{1/d}/\varepsilon)$.
\end{proof}

The framework recovers human--AI alignment as a special case (\Cref{tab:universality}), but the claims split into three layers. \emph{Theorem-backed:} the structural lower bound says any communication protocol must pay for the quotient mismatch. \emph{Constructive:} the logarithmic-in-accuracy upper bound comes from the memoryless/codebook regime above, not from the full general model. \emph{Qualitative:} RLHF, debate, routing, and scalable oversight are interpretations of the structural theorem once the quotient model is deemed appropriate for the application.

In RLHF~\cite{christiano2017}, each binary preference comparison conveys at most 1~bit (a choice between two options); with Likert-scale strength, the effective rate is $\approx$1--2~bits per comparison.\footnote{The 1-bit lower bound is information-theoretic: a binary choice has entropy $\leq 1$~bit. The ``1--2~bits'' range reflects that graded preferences (e.g., 5-point Likert) can convey $\log_2 5 \approx 2.3$~bits, reduced by human noise. This is a modeling assumption, not a derived quantity.} Under the quotient model, \Cref{thm:align-cost} therefore predicts a structural cost \emph{linear} in the capacity gap and an accuracy surcharge \emph{logarithmic} in~$1/\varepsilon$. These are scaling-form predictions---the theory identifies the functional dependence on the capacity gap and accuracy target, not a computable numerical bound for a specific LLM--human pair. Estimating effective quotient cardinalities for real systems (e.g., via probing classifiers that approximate the quotient partition, or via the lattice-gradient approach of \Cref{app:lattice-gradient}) is the key open problem connecting this theory to practice. The traversal theorem (\Cref{thm:traversal}) suggests routing alignment through intermediate abstractions---analogous to weak-to-strong generalization or scalable oversight via debate. Data processing implies post-processing cannot improve alignment. Extended discussion in \Cref{app:alignment-details}.

%% ============================================================
\section{Experiments}
\label{sec:experiments}

We validate theoretical predictions on \textsc{Chain5} ($|S|\!=\!5$, $|A|\!=\!2$, $|O|\!=\!5$), the standard RockSample(4,4) benchmark ($|S|\!=\!257$, $|A|\!=\!9$, $|O|\!=\!3$), and six additional environments (up to $|S|\!=\!200$). All results report median over 10~seeds; shaded bands are IQR. We distinguish two empirical roles: fixed-$\varepsilon$ figures illustrating the structural phase transition, and one dedicated shrinking-distortion sweep aligned with the asymptotic converse regime. Detailed methodology is in \Cref{app:experiments}; \Cref{app:exp-scope} records what these experiments do and do not validate relative to the theorem-level claims.

\paragraph{Methodology summary.}
Quotient classes are computed by enumerating all deterministic FSCs for $m \leq 3$ or sampling $n_{\mathrm{FSC}} \geq 50$ random stochastic FSCs for larger~$m$; histories with identical behavioral signatures (future observation distributions across all sampled FSCs, to tolerance~$10^{-4}$) form equivalence classes. Intent distortion uses the two-term proxy $\dintent^{\mathrm{exp}}(h) := \|b_h^A - b_h^B\|_1 + 0.5 \cdot \mathbf{1}\{a_A \neq a_B\}$, which lower-bounds the full three-term $\dintent$ of \Cref{sec:distortion} while preserving the same phase-transition location (\Cref{prop:proxy}). Distortion is measured as the empirical average over $10{,}000$ sampled trajectories. Semantic codebooks are constructed via $k$-means++ on quotient-class belief centroids (\Cref{def:codebook}). All code and reproduction instructions are available at \url{https://github.com/alch3mistdev/semantic-rate-distortion}.

\begin{figure}[t]
\centering
\includegraphics[width=0.48\columnwidth]{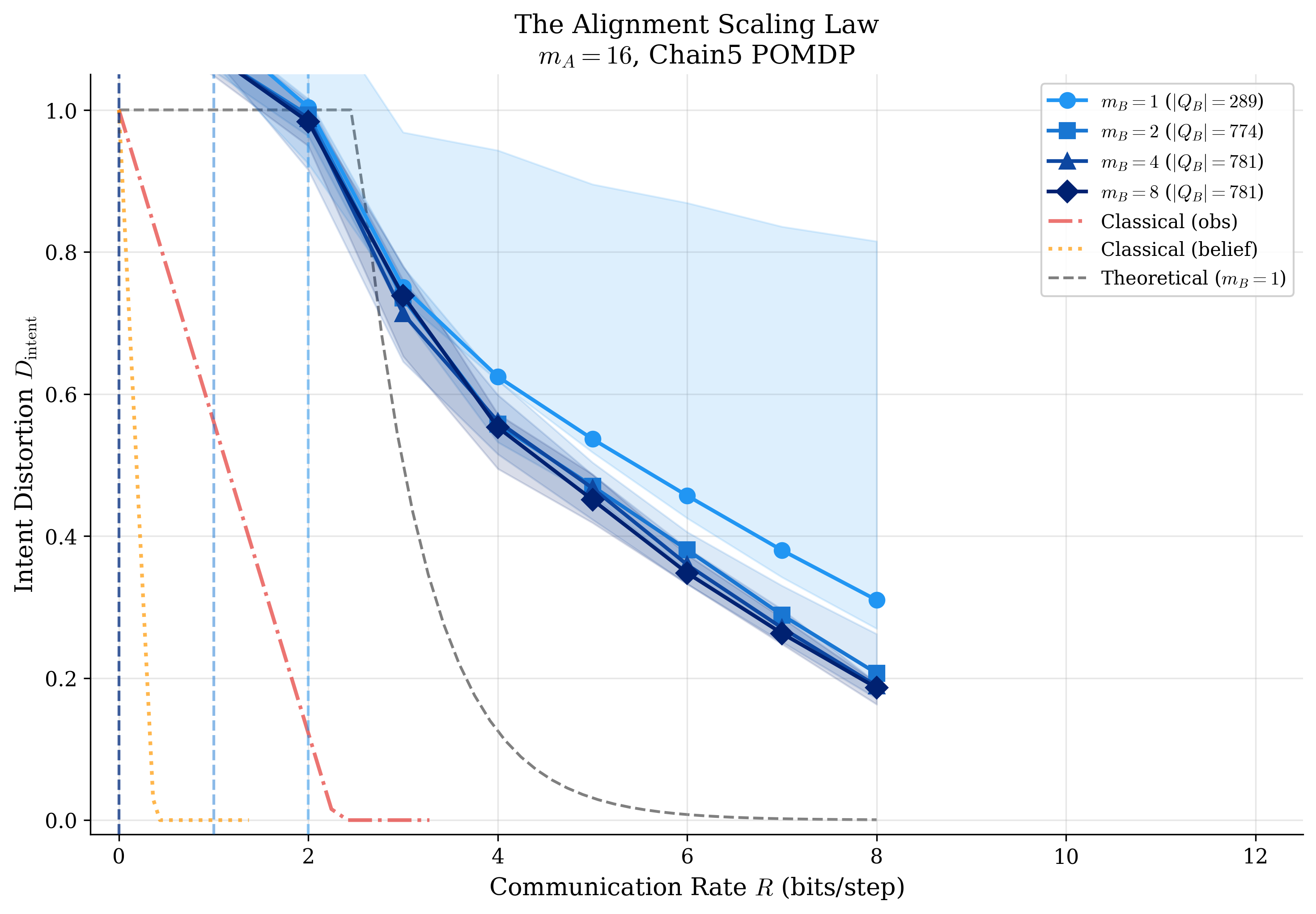}
\includegraphics[width=0.48\columnwidth]{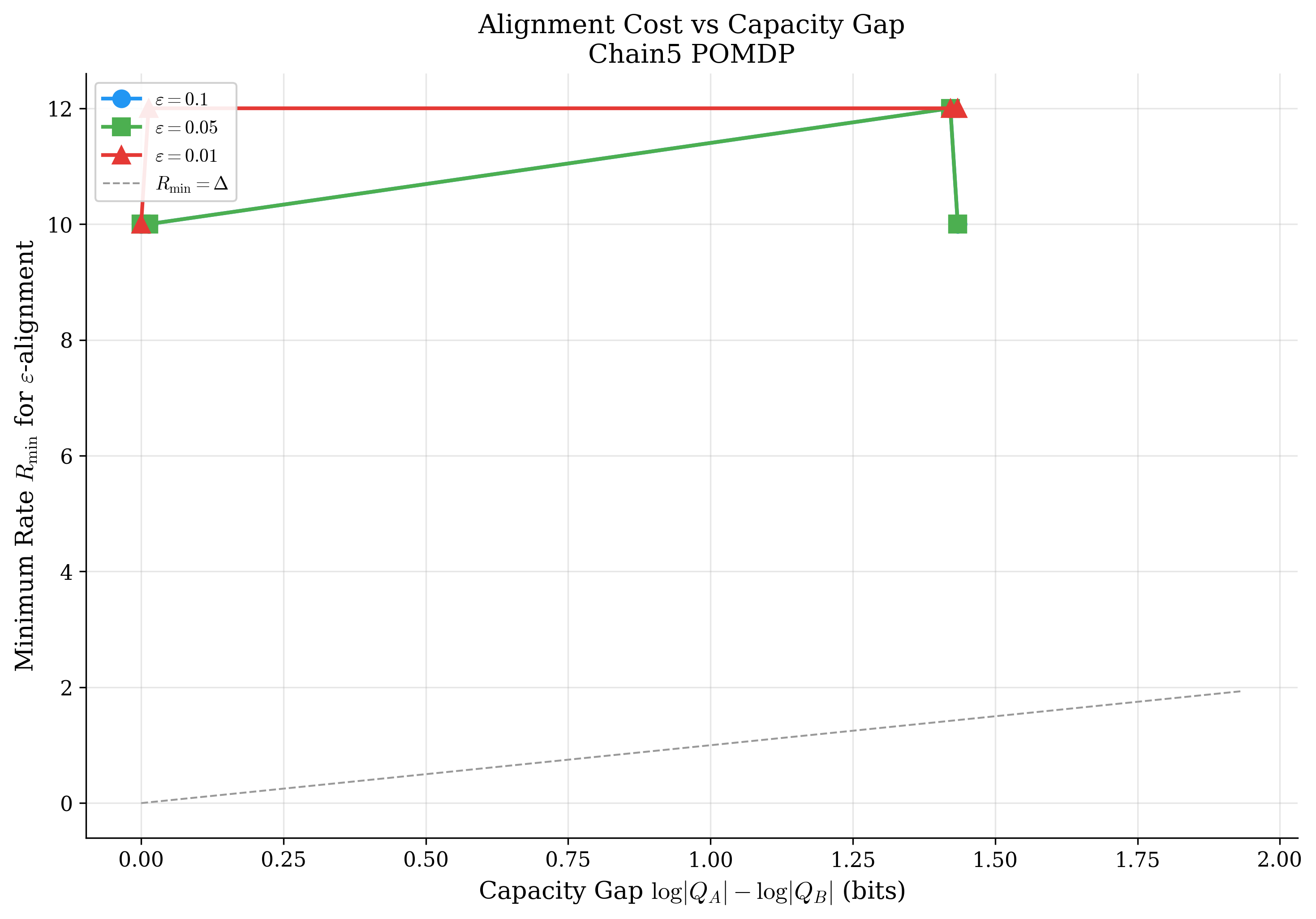}
\caption{\textbf{Left:} Chain5 rate-distortion curves for $m_A\!=\!16$, $m_B \in \{1,2,4,8\}$. A sharp knee appears near $\Rcrit \approx 1.4$~bits/step (dashed); distortion decays rapidly above it. \textbf{Right:} $R_{\min}$ vs.\ capacity gap (slope $\approx 1$), confirming \Cref{thm:align-cost}.}
\label{fig:phase-transition}
\end{figure}

\paragraph{Phase transition (\Cref{fig:phase-transition}, left).}
The $m_B\!=\!1$ curve ($|\calQ_A|\!=\!781$, $|\calQ_B|\!=\!289$) shows a clear empirical knee near the log-cardinality reference $\approx 1.4$~bits/step; $m_B \geq 2$ curves cluster together once the capacity gap closes. Because Chain5 is highly non-uniform ($\max_k r_k / \bar{r} = 15.9$; \Cref{app:experiments}), we treat this as \emph{structural evidence for a knee near the counting reference}, not as validation of a sharper theorem-level exponent.

\paragraph{Alignment scaling (\Cref{fig:phase-transition}, right).}
$R_{\min}$ scales linearly with the capacity gap with slope $\approx 1$, confirming the structural term dominates (\Cref{thm:align-cost}).

\begin{figure}[t]
\centering
\includegraphics[width=0.74\columnwidth]{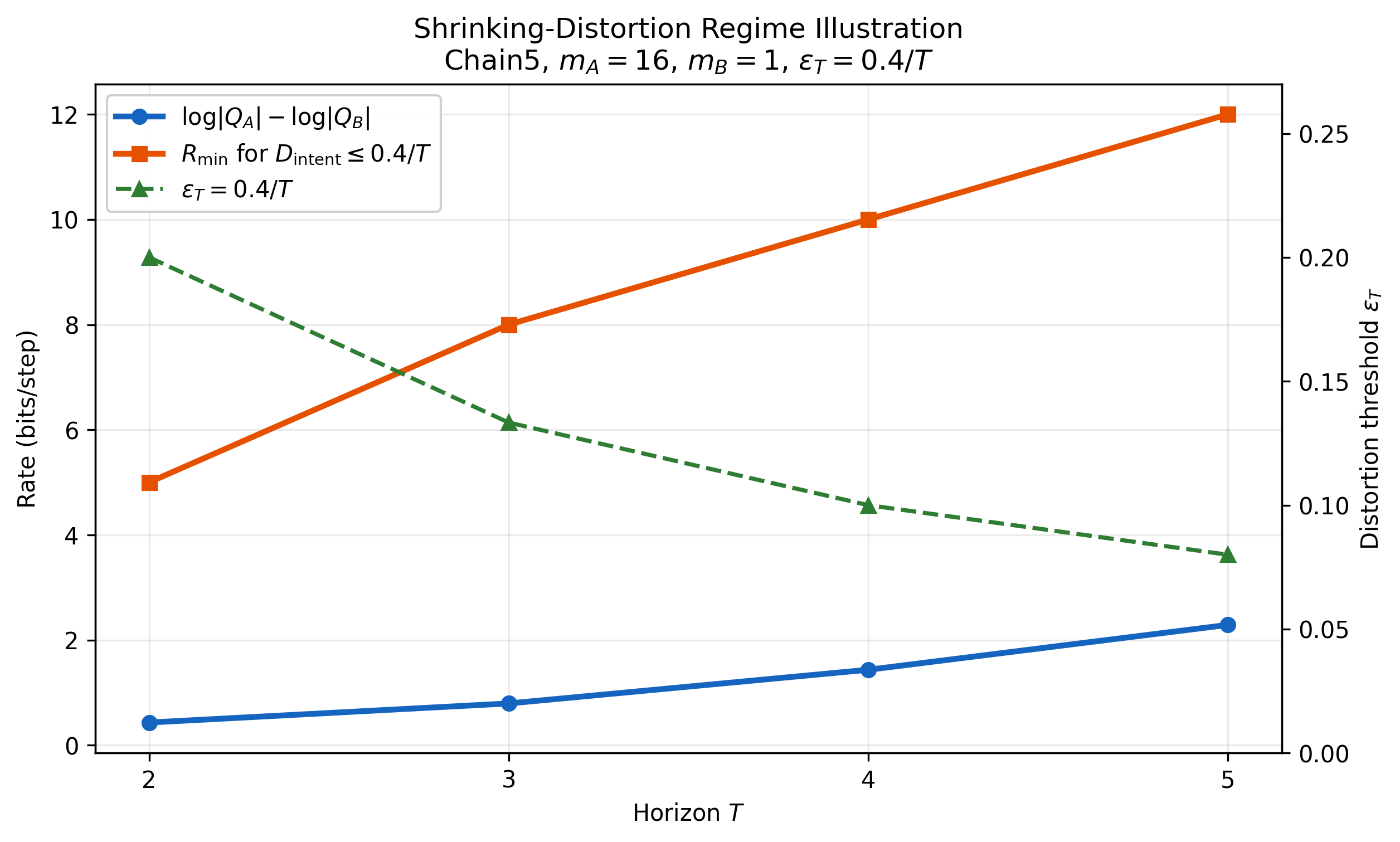}
\caption{Shrinking-distortion regime illustration for the asymptotic one-way converse (Chain5, $(m_A,m_B)=(16,1)$). The orange curve shows the empirical minimum rate $R_{\min}$ achieving $D_{\mathrm{intent}} \leq \varepsilon_T$ with $\varepsilon_T = 0.4/T$; the blue curve shows the log-cardinality reference $\log|Q_A|-\log|Q_B|$. This figure is a \emph{regime-matching illustration} for \Cref{thm:converse}. By contrast, \Cref{fig:phase-transition} remains a fixed-$\varepsilon$ illustration of the broader structural phase-transition theorem.}
\label{fig:shrinking-eps}
\end{figure}

\paragraph{Shrinking-distortion regime match (\Cref{fig:shrinking-eps}).}
To complement the fixed-$\varepsilon$ plots, we run a dedicated Chain5 sweep in the same regime as \Cref{thm:converse}: for $T \in \{2,3,4,5\}$ we set $\varepsilon_T = 0.4/T$, recompute $(\calQ_A,\calQ_B)$ for each horizon, and report the minimum integer rate achieving $D_{\mathrm{intent}} \leq \varepsilon_T$. The resulting thresholds are $R_{\min} \in \{5,8,10,12\}$, while the corresponding log-cardinality references are $\{0.43, 0.79, 1.43, 2.29\}$ bits/step. We use this sweep as a regime-matching illustration rather than a pointwise empirical proof of the converse: it keeps the distortion tolerance in the same asymptotic scaling class as the theorem, while the fixed-$\varepsilon$ figures continue to illustrate the broader structural transition.

\paragraph{Scalability and closer-to-uniform cases.}
\Cref{fig:richgrid} shows the phase transition on RichGridWorld ($|S|\!=\!36$, $|O|\!=\!12$, $T\!=\!2$); the $m_B\!=\!1$ gap ($\Rcrit \approx 5.0$~bits/step) creates a persistent distortion floor. \textsc{BalancedRand8} ($|S|\!=\!8$, ratio~$2.58$; \Cref{fig:richgrid}) is much closer to the uniform/cardinality regime than Chain5 and shows an inflection near $\Rcrit \approx 1.63$~bits/step, which is consistent with the idealized constructive benchmark without being a proof of it. To verify that the phase transition persists at scale, \Cref{fig:scalability} tests chain POMDPs with $|S| \in \{100, 150, 200\}$ and coarse observations ($|O| \in \{5,6,8\}$, $T\!=\!2$, $m_A\!=\!4$). The $m_B\!=\!1$ curves show a clear capacity gap ($|\calQ_A| > |\calQ_B|$, $\Rcrit = 2$~bits/step) with distortion bounded away from zero below~$\Rcrit$; the $m_B\!=\!2$ curves close the gap ($|\calQ_A| = |\calQ_B|$, $\Rcrit = 0$). Total runtime for all three domains: $< 1$~second.

\begin{figure}[t]
\centering
\includegraphics[width=0.95\columnwidth]{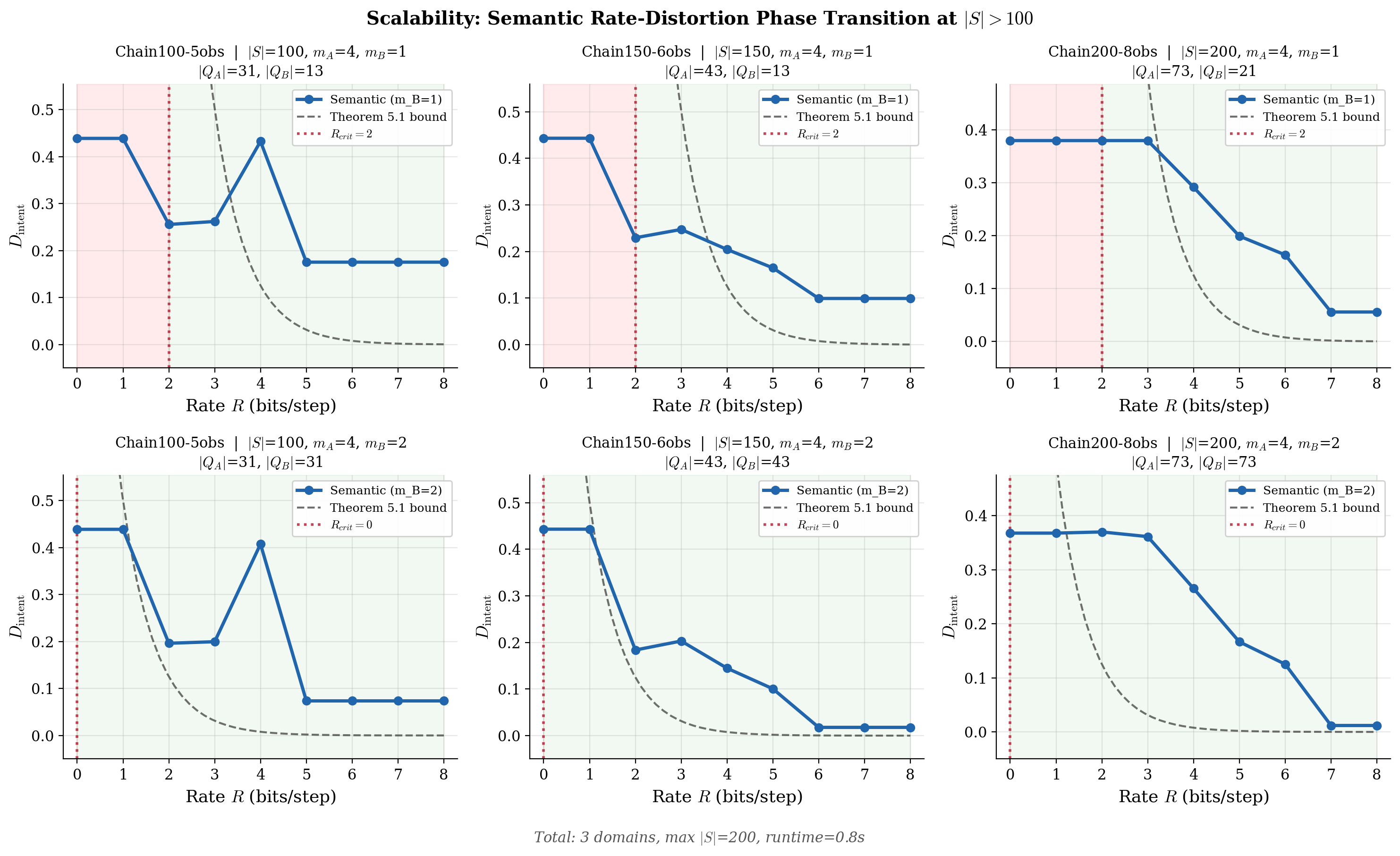}
\caption{Phase transition at $|S| > 100$. \textbf{Top:} $m_B\!=\!1$ creates a capacity gap; distortion floor persists below $\Rcrit = 2$~bits/step. \textbf{Bottom:} $m_B\!=\!2$ closes the gap ($|\calQ_A| = |\calQ_B|$). Domains: Chain100 ($|S|\!=\!100$), Chain150 ($|S|\!=\!150$), Chain200 ($|S|\!=\!200$).}
\label{fig:scalability}
\end{figure}

\paragraph{Standard benchmark: RockSample(4,4).}
To test the framework on a recognized POMDP benchmark beyond synthetic chains, we apply it to RockSample(4,4) \cite{smith2004heuristic} ($|S|\!=\!257$, $|A|\!=\!9$, $|O|\!=\!3$, $T\!=\!3$). This domain has \emph{action-dependent observation support}: only check-rock actions yield informative observations (good/bad), while movement always emits ``none.'' The quotient source is therefore non-i.i.d.: rock-checking strategies create temporal correlations in the observation sequence that make the quotient class at time~$t$ depend on the full history, not just the current state. \Cref{fig:rocksample} confirms a clear phase transition: for $m_A\!=\!8$, $m_B\!=\!1$ ($|\calQ_A|\!=\!40$, $|\calQ_B|\!=\!7$, $\Rcrit\!=\!3$), intent distortion remains above $0.73$ below the critical rate and drops to $0.52$ above it. Larger receiver capacity ($m_B\!=\!2,3$) progressively lowers the distortion floor ($0.17$ and $0.11$, respectively). The Blahut-Arimoto $R_{\mathrm{WZ}}(D)$ benchmark yields $H(\calQ_A|\calQ_B) = 0.70$~bits/step---confirming non-trivial side-information gain on a non-i.i.d.\ source---and the $R_{\mathrm{WZ}}(D)$ curve lies strictly below $R(D)$ throughout the distortion range. Runtime: $< 7$~seconds.

\begin{figure}[t]
\centering
\includegraphics[width=0.95\columnwidth]{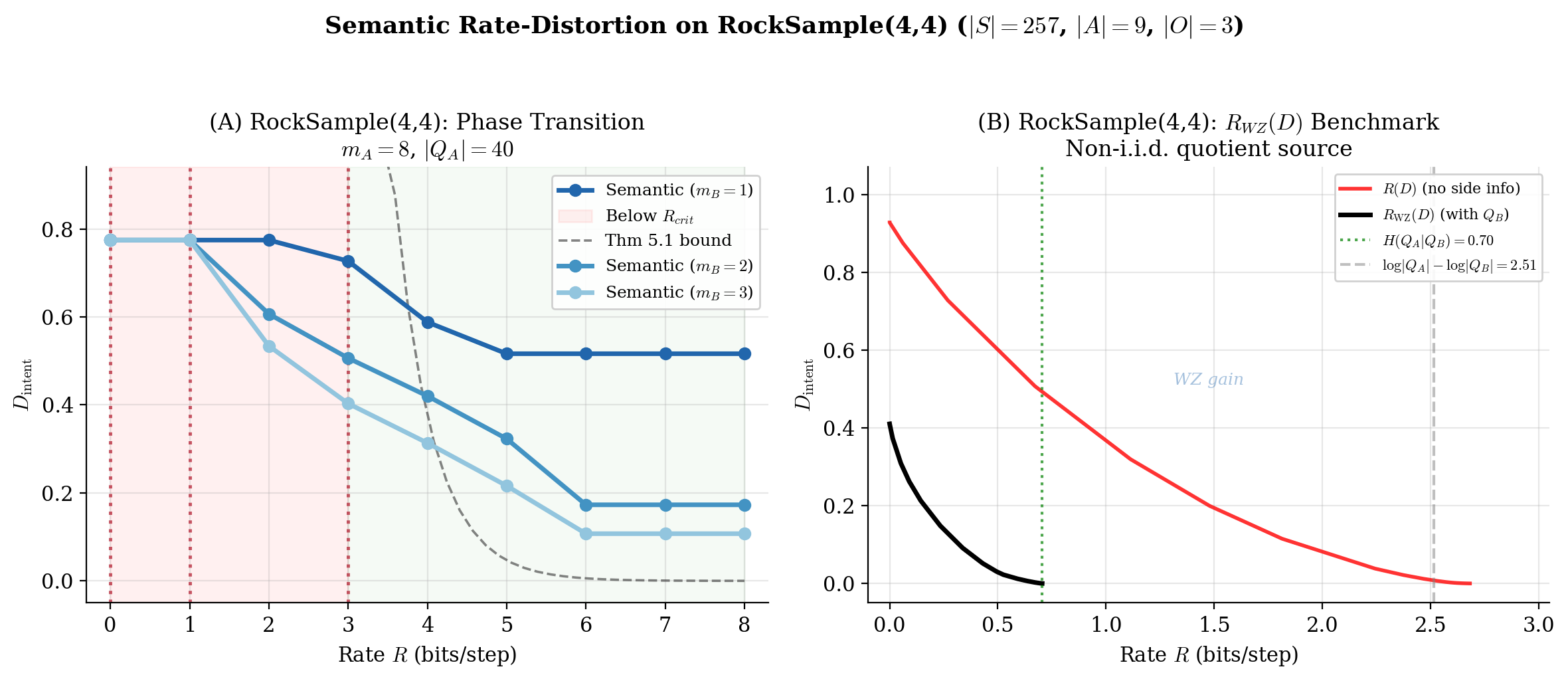}
\caption{RockSample(4,4): phase transition and WZ benchmark. \textbf{Left:} $D_{\mathrm{intent}}$ vs.\ $R$ for $m_A\!=\!8$, $m_B \in \{1,2,3\}$; the $m_B\!=\!1$ curve shows a clear knee at $\Rcrit\!=\!3$~bits/step. \textbf{Right:} Blahut-Arimoto $R_{\mathrm{WZ}}(D)$ (black) vs.\ $R(D)$ (red) on the non-i.i.d.\ quotient source; $H(\calQ_A|\calQ_B) = 0.70$~bits/step.}
\label{fig:rocksample}
\end{figure}

\begin{figure}[t]
\centering
\includegraphics[width=0.48\columnwidth]{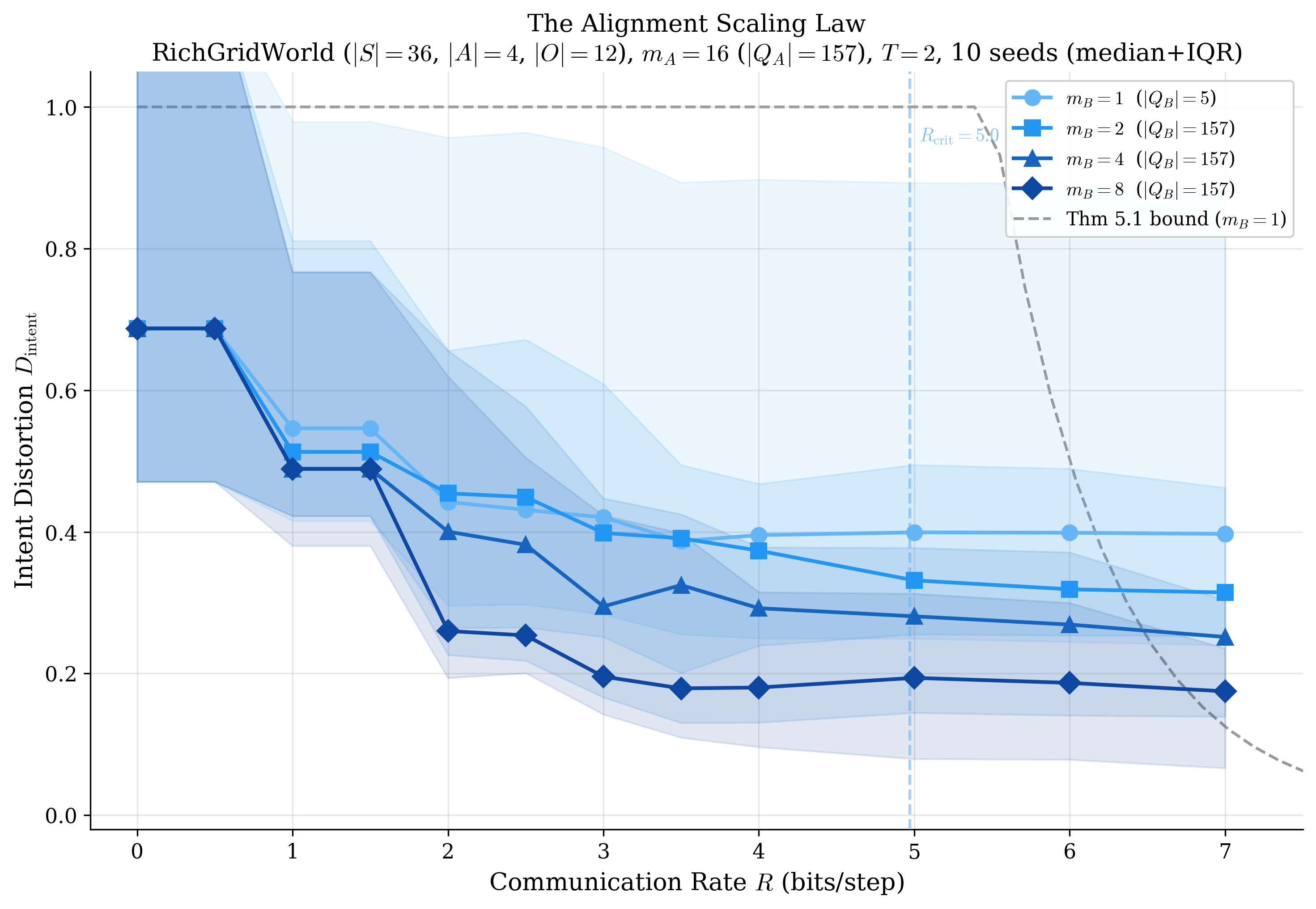}
\includegraphics[width=0.48\columnwidth]{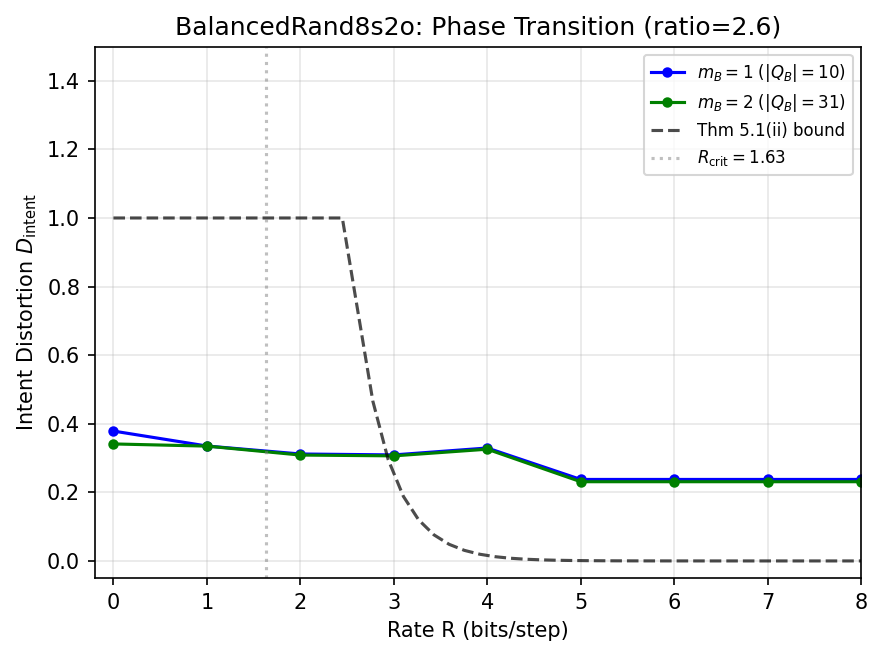}
\caption{\textbf{Left:} RichGridWorld ($T\!=\!2$); $m_B\!=\!1$ gap creates a distortion floor. \textbf{Right:} BalancedRand8 ($T\!=\!4$, refinement ratio~$2.58$); inflection near $\Rcrit \approx 1.63$.}
\label{fig:richgrid}
\end{figure}

\paragraph{Baselines and semantic coding advantage.}
Against Blahut-Arimoto $R_{\mathrm{WZ}}(D)$, k-means on beliefs, and random clustering (\Cref{app:experiments}), the quotient-aware protocol consistently outperforms; semantic coding requires higher rates than classical observation compression, reflecting the cost of intent preservation (\Cref{fig:semantic-vs-classical}). Additional experiments (Tiger, LLM routing) are in the appendix.

\paragraph{Entropy-rate benchmark vs.\ log-cardinality (\Cref{fig:structured-rd}).}
The experiments above use random policies, which produce near-uniform quotient visitation; the empirical knee therefore tracks the \emph{log-cardinality bound} ($\log|\calQ_A| - \log|\calQ_B| = 1.43$~bits/step). \Cref{cor:exact-rcrit} identifies the sharper one-way benchmark as the conditional entropy rate $H(\calQ_A \mid \calQ_B)$, which equals the log-cardinality form only under uniform visitation. To probe this sharper benchmark, we compare Blahut-Arimoto Wyner-Ziv $R_{\mathrm{WZ}}(D)$ curves under two source distributions:
\begin{enumerate}[nosep,leftmargin=*]
  \item \emph{Random policy}: uniform actions yield $H(\calQ_A \mid \calQ_B) = 0.29$~bits/step.
  \item \emph{Structured policy}: near-optimal FSCs (m$\leq$3, value-weighted sampling) concentrate visitation, yielding $H(\calQ_A \mid \calQ_B) = 0.08$~bits/step.
\end{enumerate}
In both cases, $R_{\mathrm{WZ}}(D)$ reaches zero distortion at the predicted $H(\calQ_A \mid \calQ_B)$ benchmark rather than at the log-cardinality bound. The structured-policy knee is $3.6\times$ lower than the random-policy knee ($0.08$ vs.\ $0.29$~bits/step) at the same WZ benchmark level, and $19\times$ lower than the worst-case counting bound ($0.08$ vs.\ $1.43$~bits). The $19\times$ figure compares the tightest benchmark against the loosest; the $3.6\times$ figure isolates the effect of structured visitation at a common theoretical level. Both demonstrate that structured visitation can make the one-way benchmark substantially smaller than a capacity-counting argument suggests. This is benchmark evidence, not a separate proof of full end-to-end operational optimality outside the regimes covered by \Cref{thm:wz}.

\begin{figure}[t]
\centering
\includegraphics[width=0.75\columnwidth]{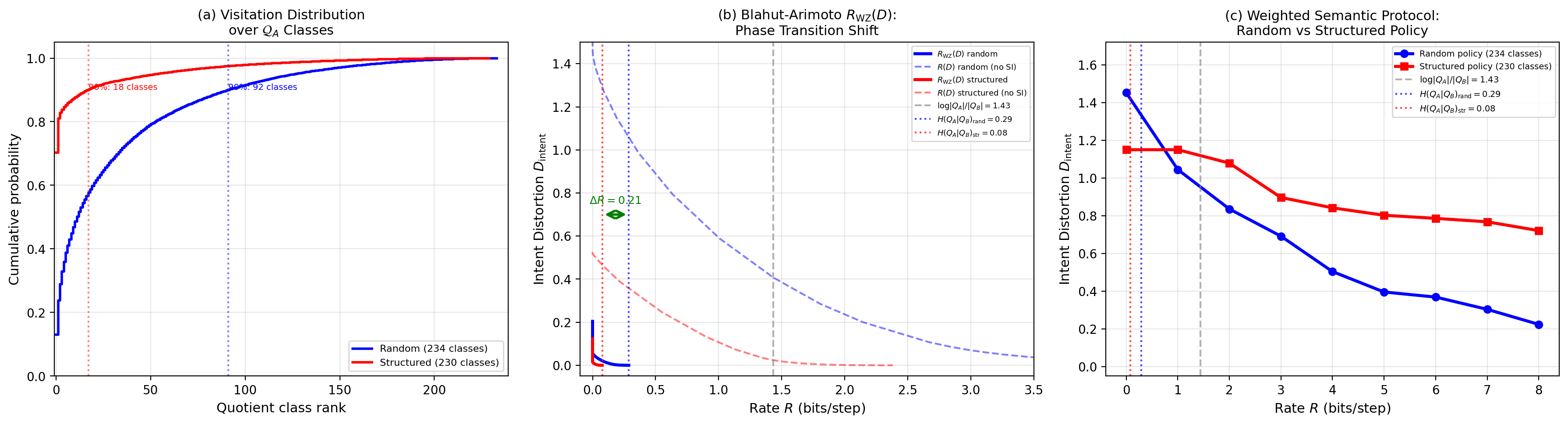}
\caption{Wyner-Ziv benchmark curves under random vs.\ structured policies (Chain5, $m_A\!=\!16$, $m_B\!=\!1$). Solid: $R_{\mathrm{WZ}}(D)$ with decoder side information; dashed: $R(D)$ without. Vertical lines mark the log-cardinality bound ($1.43$, gray), $H(\calQ_A|\calQ_B)_{\mathrm{rand}}$ ($0.29$, blue), and $H(\calQ_A|\calQ_B)_{\mathrm{struct}}$ ($0.08$, red). Structured vs.\ random at the same WZ benchmark level: $3.6\times$ reduction ($0.08$ vs.\ $0.29$); structured vs.\ worst-case counting bound: $19\times$ ($0.08$ vs.\ $1.43$).}
\label{fig:structured-rd}
\end{figure}

\paragraph{Encoder-only bottlenecks vs.\ semantic WZ (\Cref{fig:ib-comparison}).}
\Cref{prop:ib-subsumption} shows that any encoder-only bottleneck pays at least the WZ benchmark plus the relevance it retains about~$Q_B$. The figure compares an IB-style encoder-only baseline against the semantic WZ benchmark. At the lossless endpoint, the gap is exact: under structured policies, encoder-only lossless coding requires $H(\calQ_A) = 2.39$~bits/step while the one-way semantic benchmark needs only $H(\calQ_A \mid \calQ_B) = 0.08$~bits/step---a $31\times$ gap. Under random policies the corresponding lossless ratio is $21\times$. Across the displayed distortion range, the encoder-only curve remains above the side-information-aware benchmark, illustrating the cost of ignoring decoder side information.

\begin{figure}[t]
\centering
\includegraphics[width=0.95\columnwidth]{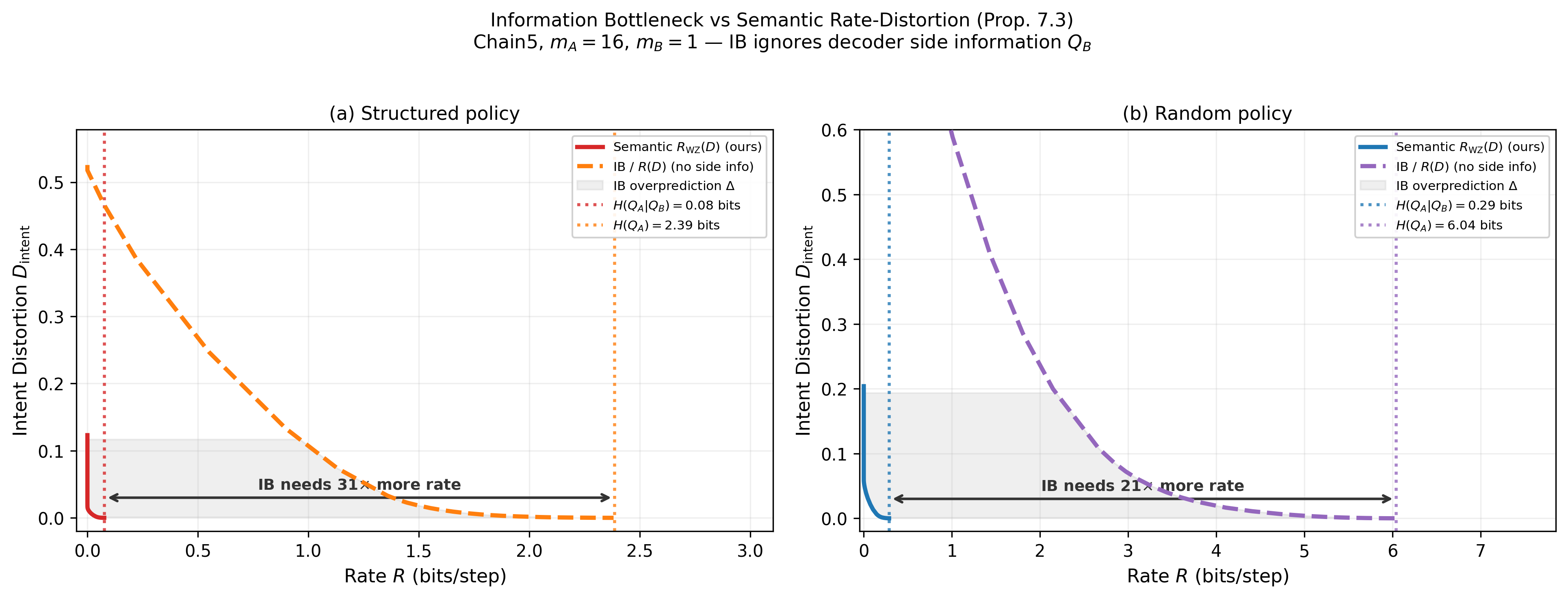}
\caption{IB-style encoder-only compression vs.\ the semantic WZ benchmark (\Cref{prop:ib-subsumption}). The dashed curve ignores decoder side information~$Q_B$; the solid curve exploits it. The exact theorem-level statement is the lower bound of \Cref{prop:ib-subsumption}; the lossless endpoints differ by $H(Q_B)$. \textbf{(a)}~Structured policies: $31\times$ lossless gap. \textbf{(b)}~Random policies: $21\times$ lossless gap.}
\label{fig:ib-comparison}
\end{figure}

%% ============================================================
\section{Related Work}
\label{sec:related}

\paragraph{Rate-distortion and source coding with side information.}
Shannon~\cite{shannon1948,shannon1959} established rate-distortion for shared codebooks; the Blahut-Arimoto algorithm~\cite{blahut1972,arimoto1972} provides the computational backbone. Wyner and Ziv~\cite{wynerziv1976} characterized rate-distortion with decoder side information; Slepian and Wolf~\cite{slepianwolf1973} established the complementary lossless distributed coding result. Draper and Wornell~\cite{draper2004} extended WZ to structured side information, closely related to our quotient side-information model. Our one-way WZ reduction (\Cref{thm:wz}) identifies when semantic rate-distortion can be benchmarked by a classical side-information problem on quotient alphabets. Derpich and {\O}stergaard~\cite{derpich2012} improve causal rate-distortion bounds; Permuter et al.~\cite{permuter2008trapdoor} establish directed information achievability. Stavrou and Kountouris~\cite{stavrou2022isit,stavrou2023task} and Zaidi et al.~\cite{zaidi2020} develop goal-oriented compression where the distortion is task-specific; our framework differs in that both the source alphabet (quotient classes) and the distortion (intent preservation) are \emph{derived} from agent capacity rather than pre-specified.

\paragraph{Semantic communication, bounded rationality, and Dec-POMDPs.}
G\"und\"uz et al.~\cite{gunduz2023}, Kountouris and Pappas~\cite{kountouris2021}, and Xie et al.~\cite{xie2021} develop task-oriented and deep-learning-enabled semantic communication. Tishby et al.~\cite{tishby1999} introduced the Information Bottleneck (IB), which compresses through a pre-specified relevance variable~$Y$; \Cref{prop:ib-subsumption} shows that once $Y = Q_B$ is fixed, any encoder-only bottleneck pays an additive penalty relative to the side-information-aware WZ benchmark, and at zero distortion that penalty is exactly~$H(Q_B)$. Crucially, in our framework~$Y$ is \emph{derived} from capacity via the quotient functor, not chosen. Sims~\cite{sims2003} and Genewein et al.~\cite{genewein2015} formalized information-theoretic bounded rationality; Ortega and Braun~\cite{ortegabraun2013} developed a thermodynamic framework. Our framework grounds these in POMDPs via the quotient functor. Bernstein et al.~\cite{bernstein2002} showed Dec-POMDP planning is NEXP-complete; Goldman and Zilberstein~\cite{goldmanzilberstein2004} characterized communication complexity. Nayyar, Mahajan, and Teneketzis~\cite{nayyar2013} developed the \emph{common-information approach} to Dec-POMDPs under communication constraints, establishing structural results for optimal strategies given shared information; our one-way observability (\Cref{ass:one-way}) creates a related asymmetric information structure, but our focus is on minimum \emph{rate} for semantic alignment rather than optimal \emph{strategies}. Tatikonda and Mitter~\cite{tatikonda2004} show stabilization requires rate exceeding topological entropy; $\Rcrit$ is the analogous semantic quantity. Witsenhausen~\cite{witsenhausen1968} shows shared-information assumptions break distributed control; \Cref{thm:capacity}(i) is a semantic analogue. In multi-agent RL, emergent communication protocols arise spontaneously when agents are given discrete channels~\cite{lazaridou2017,eccles2019}; our framework provides a theoretical lens for such protocols, predicting the minimum channel capacity for intent-preserving coordination as a function of the agents' capacity gap.

\paragraph{POMDP abstraction and alignment.}
Nixon~\cite{nixon2026} establishes the Myhill--Nerode theorem for bounded interaction; the quotient POMDP's well-definedness is sketched in \Cref{rem:quotient-welldef}. Castro et al.~\cite{castro2009}, Ferns et al.~\cite{ferns2004}, Li et al.~\cite{li2006}, Abel et al.~\cite{abel2019}, and Amato et al.~\cite{amato2010} develop POMDP/MDP equivalence, abstraction, and FSC-based algorithms. The quotient functor~$Q$ unifies bisimulation, lumpability, and policy abstraction as capacity-indexed special cases~\cite{nixon2026}. Prior alignment theory~\cite{russell2019,soares2014} focuses on preference learning; Christiano et al.~\cite{christiano2017} introduced deep RL from human preferences (RLHF), the dominant paradigm for practical alignment. We provide a complementary information-theoretic perspective: RLHF's comparison budget is bounded below by $\Rcrit$ (\Cref{sec:applications}).

%% ============================================================
\section{Discussion and Conclusion}
\label{sec:discussion}

We have developed a semantic rate-distortion theory where bounded agents' semantic spaces are quotient POMDPs and communication is formalized as rate-constrained quotient morphisms. The backbone of the paper is the fixed-$\varepsilon$ structural phase transition theorem: below~$\Rcrit$, some semantic distinctions are irresolvable regardless of protocol design. Under one-way observability in the common-history coarsening regime, the Wyner-Ziv benchmark then sharpens this structural picture: the converse matches the classical WZ converse exactly, the operational bridge is exact for memoryless quotient sources (\Cref{cor:iid-achievability}), the resulting constructive decay is classical rather than bespoke, and the ergodic bridge is argued with explicit mixing-time bounds for long-run average distortion (\Cref{rem:wz-open}). Separately, the shrinking-distortion converse provides an asymptotic one-way lower bound from messages plus decoder side information rather than the sole empirical backbone of the paper.

The impossibility below~$\Rcrit$ is \emph{structural}: when the sender's quotient resolution exceeds the receiver's, certain semantic distinctions are irresolvable regardless of coding sophistication. The discrete framework extends to continuous domains via metric entropy: $\Rcrit(\varepsilon) \approx (d_A - d_B) \log(1/\varepsilon)$ for smooth quotient manifolds of dimensions $d_A > d_B$ (\Cref{app:continuous}).

\paragraph{Summary of results and proof status.}
\Cref{tab:proof-status} collects the main results with their proof status and required assumptions. \Cref{app:validation-ledger} separately records which empirical sections speak to which claims, so theorem-level, benchmark-level, and qualitative statements are not conflated.

\begin{table}[t]
\centering
\footnotesize
\begin{tabularx}{\linewidth}{@{}p{0.22\linewidth}p{0.23\linewidth}Yp{0.15\linewidth}@{}}
\toprule
\textbf{Result} & \textbf{Status} & \textbf{Assumptions} & \textbf{Loc.} \\
\midrule
Value/proxy bounds & Proved & Lip.\ reward & \S\ref{sec:distortion} \\
Capacity (i): threshold & Proved & Common-history coarsening + positive support on a merged class & App.~\ref{app:proof-51i} \\
Capacity (ii): exp.\ decay & Proved (i.i.d.\ one-way) & Common-history coarsening, intent meas., one-way, i.i.d.\ quotient source, Lip. & App.~\ref{app:proof-51ii} \\
Shrinking-distortion converse & Proved (shrinking-$\varepsilon$ regime) & Common-history coarsening, intent meas., quot.\ reg., one-way & App.~\ref{app:proof-61} \\
WZ converse (Prop.~\ref{thm:wz}(i)) & Proved & Common-history coarsening, intent meas.\ + one-way & \S\ref{sec:wynerziv} \\
WZ bridge (Prop.~\ref{thm:wz}(ii)) & Proved (i.i.d.) / Argued (erg.) & Common-history coarsening, intent meas.\ + one-way & \S\ref{sec:wynerziv} \\
$\Rcrit$ benchmark / gap closure & Exact converse / exact i.i.d.; argued ergodic & Common-history coarsening, intent meas.\ + one-way & \S\ref{sec:wynerziv} \\
Traversal & Proved & Lip.\ morph. & App.~\ref{app:proof-traversal} \\
Single-letter & Proved & Memoryless~$\mu$ & \S\ref{sec:achievability} \\
Alignment scaling & Lower bound proved; upper bound constructive & Cap.\ thm + memoryless/codebook regime & \S\ref{sec:applications} \\
Encoder-only penalty (Prop.~\ref{prop:ib-subsumption}) & Proved & Common-history coarsening, intent meas.\ + one-way & \S\ref{sec:wynerziv} \\
Quotient PAC bound (Prop.~\ref{prop:pac-quotient}) & Proved & Separation $\gamma > 0$ & App.~\ref{app:experiments} \\
\bottomrule
\end{tabularx}
\caption{Proof status. ``Proved'' = complete proof in appendix. ``Proved (i.i.d.) / Argued (erg.)'' = exact for memoryless sources; ergodic case argued with explicit mixing-time bounds (\Cref{rem:wz-open}). ``Argued'' = structured argument; gap acknowledged. The alignment upper bound is constructive rather than fully general, and the shrinking-distortion converse is theorem-level only in its stated asymptotic regime.}
\label{tab:proof-status}
\end{table}

\paragraph{On the capacity model and the source of novelty.}
The framework models agent capacity as the node count~$m$ of a finite-state controller. This is one possible formalization---not the only one. Attention-limited agents, agents with context-dependent working memory, or neural networks with varying depth and width all suggest alternative capacity measures that would induce different quotient structures and potentially different~$\Rcrit$ values. We make two observations. First, the \emph{structural prediction}---that a capacity gap induces a phase transition in semantic communication---is robust to the choice of capacity model: any model that produces a refinement lattice of abstractions (finer capacity $\Rightarrow$ finer partition) yields a qualitatively identical phase transition, with $\Rcrit$ determined by the refinement gap. The FSC formalization is the setting in which we can prove sharp theorems, but the structural phenomenon does not depend on it. Second, the framework is \emph{modular}: replacing the FSC-based quotient with any abstraction functor that satisfies right-invariance (\Cref{lem:quotient-rightinv}) and refinement monotonicity (\Cref{lem:refinement}) would preserve all results from \Cref{sec:capacity} onward. The FSC assumption determines \emph{which} quotient is computed; the rate-distortion theory downstream is parametric in the quotient structure.

One might separately object that the framework reduces to ``apply standard coding theorems to a particular alphabet.'' This objection conflates the contribution with its consequences. Classical rate-distortion theory tells you the communication cost \emph{once you know the source alphabet}. But in the multi-agent setting, the alphabet is not given---it emerges from each agent's computational constraints. The quotient functor~$Q$ identifies which environmental distinctions each agent's memory can sustain; the refinement gap between $\calQ_A$ and $\calQ_B$ is what creates the communication problem in the first place. That $\Rcrit$ is metric-invariant (\Cref{prop:proxy}) is not a weakness but a feature: the phase transition is a \emph{structural} property of the capacity gap, independent of how we measure semantic fidelity.

\paragraph{Limitations and open problems.}
The WZ reduction requires both the common-history coarsening regime and one-way observability; the two-way and genuinely distinct-sensor cases remain open~\cite{tatikonda2004,permuter2009} (though \Cref{rem:distinct-sensor} identifies a class of distinct-sensor models that reduce to the common-history setting).
The exponential bound (\Cref{thm:capacity}(ii)) is proved only in the one-way memoryless regime; outside that regime we present structural lower bounds and benchmark identifications rather than a theorem-level constructive exponent.
Continuous alphabets, unknown environments, noisy channels, and operational WZ codebook construction are deferred to future work (\Cref{app:continuous,app:two-way}).
The alignment implications (\Cref{sec:applications}) predict scaling forms but not computable numerical bounds for real systems, pending methods for estimating effective quotient cardinalities at scale.

A concrete open problem is \emph{value-relevant tightening}. The current bounds count \emph{all} quotient distinctions equally, but some $\calQ_A$-subclasses may be value-irrelevant: merging them would not change the optimal policy or its value. Let $\calQ_A^{\mathrm{val}}$ be the coarsening of~$\calQ_A$ retaining only value-relevant cells (cf.\ \Cref{rem:conservative-bound}). Then $H(\calQ_A^{\mathrm{val}} \mid \calQ_B) \leq \bar{h}(\calQ_A \mid \calQ_B)$, and the true minimum alignment rate could be substantially lower than the full-quotient benchmark. Characterizing $\calQ_A^{\mathrm{val}}$---which requires identifying which abstractions are task-irrelevant, not merely capacity-irrelevant---would tighten all downstream bounds and is the most promising route to practical relevance. A natural candidate approach is \emph{reward-weighted partition refinement}: starting from $\calQ_A$, iteratively merge subclass pairs whose value functions differ by less than a tolerance~$\tau$, retaining only splits that change the optimal action or shift the value by more than~$\tau$. The resulting coarsening is reward-scale-aware by construction and could be combined with the lattice-gradient estimation below to make the tightening tractable at scale.

Quotient estimation has PAC-style guarantees (\Cref{prop:pac-quotient}), but the worst-case sample complexity is exponential in~$m$ (the uniform FSC distinguishing probability $p_\gamma$ can be as small as $1/N_{\mathrm{FSC}} = |\calA|^{-m} \cdot m^{-m|\calO|}$). We empirically observe rapid convergence at moderate $n \approx 20$ for structured POMDPs, and a promising path to scalability is \emph{lattice-gradient estimation} (\Cref{app:lattice-gradient}): instead of computing~$Q$ absolutely, enter the quotient lattice at a computable point and estimate the \emph{differential} refinement~$\Delta Q$ between adjacent capacity levels, chaining local estimates via Lipschitz bounds (\Cref{thm:traversal}) to recover global structure. For language-based agents, the self-referential closure of natural language provides a natural probe family---linguistic prompts that test whether a model distinguishes context~$A$ from context~$B$---making each $\Delta Q$ estimation polynomial in sample size rather than exponential in~$m$.

%% ============================================================

%% ============================================================
\appendix
\renewcommand{\theHfigure}{app.\arabic{figure}}
\renewcommand{\theHtable}{app.\arabic{table}}

%% ============================================================
%% ============================================================
\section{Wyner-Ziv Corollaries}
\label{app:wz-corollaries}

\begin{corollary}[Benchmark Gap Closure]
\label{cor:gap-closure}
\label{cor:markov-mu}
Under the one-way reduction (\Cref{thm:wz}), the converse and pipelined achievability are governed by the same WZ benchmark. For i.i.d.\ blocks drawn from the stationary marginal~$p(q_A)$, the WZ single-letter formula applies:
\[
R_{\mathrm{WZ}}(D) \;=\; \min_{\substack{p(u \mid q_A):\\ \E[\dintent] \leq D}} \bigl[I(Q_A;\, U) - I(Q_B;\, U)\bigr],
\]
where $U$ forms the Markov chain $Q_B \to Q_A \to U$. This formula is exact for the i.i.d.\ source; for the ergodic Markov quotient process, the same objective serves as a conservative constructive upper bound when we restrict attention to memoryless auxiliaries~$U$ and ignore temporal correlation.
\end{corollary}

\begin{corollary}[Inherited Results]
\label{cor:inherited}
Under the conditions of \Cref{thm:wz}, the semantic setting inherits strong converses~\cite{csiszarkorner2011} (distortion approaches~$d_{\max}$ exponentially below $\Rsem(D)$), error exponents~\cite{csiszarkorner2011} matching the WZ reliability function, and finite-blocklength bounds~\cite{kostinaverdu2012} ($R(n, D, \varepsilon) = \Rsem(D) + \sqrt{V/n}\, Q^{-1}(\varepsilon) + O(\log n / n)$).
\end{corollary}

%% ============================================================
\section{Claim / Validation Ledger}
\label{app:exp-scope}
\label{app:validation-ledger}

\begin{table}[ht]
\centering
\footnotesize
\begin{tabularx}{\linewidth}{@{}p{0.18\linewidth}p{0.21\linewidth}p{0.22\linewidth}Y@{}}
\toprule
\textbf{Claim surface} & \textbf{Formal status} & \textbf{Assumptions} & \textbf{What the experiments speak to} \\
\midrule
Structural phase transition; one-way constructive decay & Lower bound theorem-level; exponential decay theorem-level only in the i.i.d.\ one-way regime & Common-history coarsening for quotient comparison; positive support for impossibility; intent meas.\ + one-way + i.i.d.\ + Lip.\ for constructive decay & Chain5, RichGridWorld, BalancedRand8, and Chain100/150/200 probe where the empirical knees sit relative to the stated references; Chain5 is not read as validation of a sharp constructive exponent \\
Shrinking-distortion converse & Theorem-level, asymptotic only & Common-history coarsening, intent meas., one-way observability, quotient regularity, shrinking distortion $\varepsilon = O(1/T)$ & A dedicated shrinking-$\varepsilon$ Chain5 sweep matches the converse regime; the fixed-$\varepsilon$ plots still illustrate the broader phase-transition shape rather than the converse itself \\
One-way WZ benchmark & Exact converse; i.i.d.\ operationally exact; ergodic long-run bridge argued & Common-history coarsening, intent meas.\ + one-way observability & Blahut-Arimoto curves and structured-policy comparisons validate benchmark sensitivity to visitation, not full ergodic FSC-level equality \\
Alignment applications & Lower bound theorem-backed; upper bound constructive; routing qualitative & Quotient model for the application; memoryless/codebook assumptions for the upper bound & Alignment-scaling plots support the structural trend; the LLM routing appendix is an analogy/case study, not a formal instantiation \\
\bottomrule
\end{tabularx}
\caption{Ledger separating theorem-level claims from benchmark-level and qualitative claims. The purpose is not to downplay the empirical section, but to prevent fixed-$\varepsilon$ experiments and benchmark calculations from being mistaken for proof of stronger statements than the paper actually establishes.}
\end{table}

%% ============================================================
\section{Quotient Facts Used in This Paper}
\label{app:quotient-facts}

This paper does \emph{not} re-prove the full Myhill--Nerode theorem of~\cite{nixon2026}. The later sections use only three structural facts: right-invariance of the history equivalence, refinement monotonicity as capacity increases, and the resulting canonical coarsening map $\calQ_A \to \calQ_B$. We collect those facts here so the notation used in the rate-distortion arguments is self-contained. Uniqueness and minimality of the quotient remain imported background from~\cite{nixon2026}.

\begin{lemma}[Right-invariance and well-defined quotient transitions]
\label{lem:quotient-rightinv}
If $h \equiv_{m,T} h'$, then for every observation $o \in \calO$, the extended histories satisfy $ho \equiv_{m,T} h'o$. Consequently the transition $[h] \xrightarrow{o} [ho]$ is independent of the chosen representative.
\end{lemma}

\begin{proof}
If $h \equiv_{m,T} h'$, then every controller $\pi \in \Pim$ induces the same future observation law from $h$ and $h'$. Conditioning both laws on one additional observation $o$ preserves equality by the chain rule for conditional distributions in the finite POMDP, so $ho \equiv_{m,T} h'o$. The quotient transition therefore depends only on the class~$[h]$, not on the representative history.
\end{proof}

\begin{lemma}[Refinement monotonicity in capacity]
\label{lem:refinement}
If $m_A \geq m_B$ and $T_A = T_B = T$, then $\calQ_A \preceq \calQ_B$.
\end{lemma}

\begin{proof}
Since $m_A \geq m_B$, every FSC with at most $m_B$ nodes is also an FSC with at most $m_A$ nodes, so $\Pi_{m_B,T} \subseteq \Pi_{m_A,T}$. Therefore equality of future observation laws against all controllers in $\Pi_{m_A,T}$ implies equality against all controllers in $\Pi_{m_B,T}$. Hence $h \equiv_{m_A,T} h'$ implies $h \equiv_{m_B,T} h'$, so every $\calQ_A$-class is contained in a unique $\calQ_B$-class.
\end{proof}

\begin{proposition}[Canonical coarsening map]
\label{prop:canonical-coarsening}
Under the conditions of \Cref{lem:refinement}, the map
\[
\kappa_{A \to B}: \calQ_A \to \calQ_B, \qquad \kappa_{A \to B}([h]_A) := [h]_B,
\]
is well-defined. It is surjective onto the reachable $\calQ_B$-classes.
\end{proposition}

\begin{proof}
Well-definedness follows from \Cref{lem:refinement}: if $[h]_A = [h']_A$, then $h \equiv_{m_A,T} h'$, hence $h \equiv_{m_B,T} h'$ and therefore $[h]_B = [h']_B$. Surjectivity onto reachable $\calQ_B$-classes is immediate because each reachable $\calQ_B$-class contains at least one history~$h$, and that history belongs to some $\calQ_A$-class whose image under~$\kappa_{A \to B}$ is exactly~$[h]_B$.
\end{proof}

%% ============================================================
\section{Extended Capacity Theorem Remarks}
\label{app:capacity-remarks}

\begin{remark}[Positivity of $c_M$]
\label{rem:cM-positive}
$c_M > 0$ whenever $\calQ_A$ is strictly finer than $\calQ_B$, as a consequence of the definitions. In a finite POMDP, the posterior belief $b_h = P(s_t \mid h_t)$ is a sufficient statistic for the future observation distribution under \emph{any} policy~\cite{nixon2026}: $P_M^\pi(\calO_{t+1:T} \mid h)$ depends on~$h$ only through~$b_h$. Hence if $b_h = b_{h'}$, then $h$ and $h'$ produce identical future observations under all policies---including all FSCs in $\Pi_{m_A,T_A}$---and therefore lie in the same quotient class. Contrapositively, if $[h]_A \neq [h']_A$, then $b_h \neq b_{h'}$. Since $\Wone(b, b') > 0$ for distinct distributions on a finite state space, the belief component of~$\dintent$ is strictly positive: $\dintent(\mathrm{Intent}_A(h), \mathrm{Intent}_A(h')) \geq \Wone(b_h, b_{h'}) > 0$. Finiteness of~$\calQ_A$ and $\calQ_B$ then gives $c_M \geq \min_{[h]_A \neq [h']_A} \Wone(b_h, b_{h'}) > 0$.
\end{remark}

\begin{remark}[Positive-support condition in Theorem~\ref{thm:capacity}(i)]
\label{rem:merged-support}
The impossibility proof needs one additional condition beyond strict refinement: at least one merged $\calQ_B$-class witnessing the pigeonhole argument must appear with positive stationary mass under the protocol-induced source law. In finite irreducible settings this is automatic for every reachable recurrent class; we state it explicitly because mere convergence to stationarity does not by itself guarantee positive mass on every merged class.
\end{remark}

\paragraph{Role of~$B$'s quotient.}
$B$'s computational bound ($m_B$-node FSC) means its effective information state at each step is its $\calQ_B$-class. The quotient theorem~\cite{nixon2026} establishes that $\calQ_B$ is the unique minimal abstraction preserving observation laws for all policies in~$\Pi_{m_B,T_B}$. Since~$B$'s policy is constrained to this class, two histories in the same $\calQ_B$-class produce identical conditional observation distributions under any policy~$B$ can execute. Messages from~$A$ allow~$B$ to reconfigure its FSC parameters---selecting which $m_B$-node policy to run---but not to transcend the $m_B$ memory bound. The impossibility below~$\Rcrit$ arises because~$B$'s bounded processing cannot resolve the relevant $\calQ_A$-subclasses, even with optimal use of received messages.

\paragraph{Interpretation: structural impossibility.}
The \emph{existence} of a phase transition is a structural consequence of the quotient refinement $\calQ_A \preceq \calQ_B$, which is policy-independent once the common-history comparison regime is fixed. When $\calQ_A$ is strictly finer than $\calQ_B$, certain distinctions that~$A$ can perceive are invisible to~$B$ regardless of the protocol. The \emph{location} of the most informative benchmark ($\Rcrit = \bar{h}(Q_A \mid Q_B)$ under \Cref{thm:wz}) depends on the source distribution; the log-cardinality form provides a policy-independent worst-case reference, while theorem-level exponential achievability is claimed only in the one-way memoryless regime of \Cref{thm:capacity}(ii).

%% ============================================================
\section{Extended Converse Remarks}
\label{app:converse-remarks}

\begin{remark}[Entropy rate and tightness]
\label{rem:entropy-rate}
Let $h_A := \lim_{T \to \infty} \frac{1}{T} H(Q_A^T)$ and $h_B := \lim_{T \to \infty} \frac{1}{T} H(Q_B^T)$ denote the entropy rates of the quotient processes, which exist under Assumption~\ref{ass:quotient-reg}. The shrinking-distortion converse (\Cref{thm:converse}) depends on $h_A - h_B$. The benchmark is closest to the log-cardinality reference when the quotient processes have near-maximal entropy ($h_A \approx \log|\calQ_A|$, $h_B \approx \log|\calQ_B|$), i.e.\ under approximately uniform visitation of quotient classes.
\end{remark}

%% ============================================================
\section{Full Wyner-Ziv Reduction Proof}
\label{app:proof-wz}

This appendix provides the detailed proof of \Cref{thm:wz} (one-way Wyner-Ziv reduction).

\begin{proof}[Proof of \Cref{thm:wz}]
The proof proceeds in four steps; Steps~1--2 are definitional, Step~3 is the main technical content, and Step~4 follows from~\cite{csiszarkorner2011}.

\textbf{Step~1 (Source identification).}
Under Assumptions~\ref{ass:common-history}, \ref{ass:quotient-reg}, and~\ref{ass:one-way}, the quotient process $\{Q_t^A\}_{t \geq 1}$ induced by the common history source is stationary and ergodic on the finite alphabet~$\calQ_A$. Because~$B$'s actions do not feed back into the sender-side source history, $\{Q_t^A\}$ is a well-defined source independent of the decoder.

\textbf{Step~2 (Side information identification).}
Since~$\calQ_A$ and~$\calQ_B$ are both defined on that same history space, \Cref{prop:canonical-coarsening} gives a canonical map $\kappa_{A \to B}\colon \calQ_A \to \calQ_B$ with $Q_B^T = \kappa_{A \to B}(Q_A^T)$ pointwise. Thus $Q_B^T$ serves as decoder side information correlated with the source $Q_A^T$. We do \emph{not} claim this deterministic coarsening step for a genuinely distinct-sensor model unless it is first reduced to the same common-history setting.

\textbf{Step~3 (Achievability: exact i.i.d., argued ergodic).}
Any Wyner-Ziv code for the source $Q_A^T$ with decoder side information $Q_B^T$ at distortion~$D$ under~$\dintent$ operates at rate $R_{\mathrm{WZ}}(D)$. By \Cref{ass:intent-meas}, $\dintent$ is well-defined as a function of quotient class pairs $(q_A, \hat{q}_A) \in \calQ_A \times \calQ_A$, so the WZ distortion matrix is unambiguous. Such a code can be implemented as a semantic communication protocol as follows.

\emph{Causal pipeline construction.} The protocol of \Cref{def:protocol} requires~$B$ to act at every step, while WZ coding is block-based at blocklength~$n$. We bridge this gap with a \emph{pipelined} scheme that respects causality. Divide time into blocks of~$n$ steps. During block~$k$,~$A$ accumulates quotient classes $Q_A^{(k)} := (Q_{(k-1)n+1}^A, \ldots, Q_{kn}^A)$ and simultaneously transmits the WZ codeword for the \emph{previous} block~$Q_A^{(k-1)}$ at rate~$R$ bits/step. Upon receiving the codeword at the end of block~$k$,~$B$ decodes $\hat{Q}_A^{(k-1)}$ and uses it to select FSC parameters $\theta^{(k+1)}$ for block~$k+1$. During block~1 (before any codeword is available),~$B$ runs a default FSC~$\theta_0$; the resulting distortion~$D_0 \leq d_{\max}$ is bounded. Over~$K$ blocks, the average distortion is
\[
\bar{D}_K = \frac{D_0 + (K-1) D_{\mathrm{WZ}}}{K} \;\xrightarrow{K \to \infty}\; D_{\mathrm{WZ}} \leq D,
\]
so the pipeline overhead vanishes asymptotically. The per-step rate remains~$R$; the cost is a one-block decoding delay, which is standard in block coding~\cite{cover2006}.

\emph{From WZ reconstruction to semantic distortion.} The standard WZ achievability theorem for stationary ergodic sources~\cite{csiszarkorner2011} guarantees the existence of a block code at rate~$R_{\mathrm{WZ}}(D)$ whose \emph{time-averaged} reconstruction distortion converges to~$\leq D$ as blocklength~$n \to \infty$. The pipeline applies this code to successive blocks: during block~$k{+}1$, $B$ runs an FSC with parameters~$\theta^{(k+1)}$ chosen from the decoded block~$k{-}1$.

\emph{I.I.D.\ case (exact).} When the quotient source is memoryless (i.i.d.\ blocks), the parameter $\theta^{(k+1)} = f(\hat{Q}_A^{(k-1)}, Q_B^{(k-1)})$ is independent of $Q_A^{(k+1)}$, and each block's distortion equals $D_{\mathrm{WZ}}$ exactly. No mixing argument is needed, and the pipeline achieves the WZ rate-distortion function with equality: $\Rsem(D) = R_{\mathrm{WZ}}(D)$. This proves \Cref{cor:iid-achievability}.

\emph{Ergodic case (argued with mixing bounds).} For a correlated ergodic source, the per-block distortion may fluctuate because~$\theta^{(k+1)} = f(\hat{Q}_A^{(k-1)}, Q_B^{(k-1)})$ is correlated with~$Q_A^{(k+1)}$ through the source dependence. Under Assumption~\ref{ass:quotient-reg}, the quotient process on a finite state space has a unique stationary distribution and exhibits geometric mixing: $\|\Prob(Q_t^A \in \cdot \mid Q_0^A = q) - \mu\|_{\TV} \leq C_0 \rho^t$ for some $\rho < 1$ and $C_0 > 0$. The two-block delay in the pipeline creates a gap of $2n$ steps between the data used to select~$\theta^{(k+1)}$ and the block~$Q_A^{(k+1)}$ it governs. The correlation between these decays as $C_0 \rho^{2n}$. For any target tolerance $\delta > 0$, choosing blocklength
\[
n \;\geq\; \frac{2\log(C_0/\delta)}{\log(1/\rho)}
\]
ensures the per-block distortion deviation from $D_{\mathrm{WZ}}$ is at most $\delta$. Combined with the ergodic theorem applied to the joint process $(Q_A^{(k)}, \hat{Q}_A^{(k)})$, the time-averaged semantic distortion $\frac{1}{K}\sum_{k=1}^K \dintent^{(k)}$ converges almost surely to $\E[\dintent(Q_A, \hat{Q}_A)] \leq D$. The pipeline overhead from block~1 vanishes as $K \to \infty$, yielding $\bar{D}_K \leq D_{\mathrm{WZ}} + \delta + D_0/K$. This supplies an asymptotic-average achievability bridge from WZ reconstruction distortion to semantic distortion. Upgrading to a blockwise identification would require an explicit coupling that we do not provide.

\emph{Parameter optimization.} At the end of block~$k$,~$B$ has decoded $\hat{Q}_A^{(k-1)}$ and observed $Q_B^{(k-1)}$. It selects FSC parameters $\theta^{(k+1)}$ to minimize expected distortion given the decoded information: $\theta^{(k+1)} = \arg\min_{\theta \in \Theta_{m_B}} \E[\dintent \mid \hat{Q}_A^{(k-1)}, Q_B^{(k-1)}, \theta]$. The parameter space $\Theta_{m_B} = \Delta(\calA)^{m_B} \times \Delta(\{1,\ldots,m_B\})^{m_B \times |\calO|}$ (\Cref{def:protocol}) is a product of simplices, hence compact. The expected distortion is continuous in~$\theta$: the belief component~$b_h^B$ and policy component~$\pi_B(\cdot|h)$ are continuous functions of the FSC parameters (via the finite forward recursion), and $\dintent$ is continuous in beliefs and policies. By the extreme value theorem, the minimum is attained. Since~$B$ applies the optimized~$\theta^{(k+1)}$ \emph{causally} during block~$k{+}1$ (using only past decoded information), every action respects the protocol's per-step structure.

\textbf{Step~4 (Converse).}
Any semantic protocol $(E, D)$ at rate~$R$ achieving $\Dintent \leq D$ induces a valid Wyner-Ziv code in this common-history coarsening setting: the message sequence $M^T$ encodes $Q_A^T$,~$B$'s side information is $Q_B^T = \kappa_{A \to B}(Q_A^T)$, and the achieved distortion~$D$ is valid under~$\dintent$. Since $\{Q_t^A\}$ is stationary ergodic (Step~1), the WZ converse for stationary ergodic sources~\cite{csiszarkorner2011} applies: $R \geq R_{\mathrm{WZ}}(D)$.

Combining Steps~3--4 yields the benchmark identification. For i.i.d.\ sources, the identification is exact in both directions. For ergodic sources, the converse remains exact and the achievability is argued with explicit mixing-rate control; upgrading that final leg to an exact FSC-level pathwise identity remains open.
\end{proof}

\begin{remark}[Conservative nature of the bound]
\label{rem:conservative-bound}
The conditional entropy rate $\bar{h}(Q_A \mid Q_B)$ counts \emph{all} distinctions visible to~$A$ but not to~$B$. If some $\calQ_A$-subclasses are value-irrelevant (i.e., merging them does not change the optimal value), the true minimum rate could be lower. Formally, let $Q_A^{\mathrm{val}}$ denote the coarsening of~$\calQ_A$ retaining only value-relevant cells; then $H_\mu(Q_A^{\mathrm{val}} \mid Q_B) \leq \bar{h}(Q_A \mid Q_B)$, with equality when all $\calQ_A$-distinctions affect optimal value. This paper's bounds are therefore conservative; tightening them to the value-relevant quotient is an open problem.
\end{remark}

\begin{table}[ht]
\centering
\footnotesize
\begin{tabularx}{\linewidth}{@{}l l l Y r@{}}
\toprule
\textbf{Form of $\Rcrit$} & \textbf{Expression} & \textbf{Regime} & \textbf{Depends on} & \textbf{Chain5} \\
\midrule
Log-cardinality & $\log|\calQ_A| - \log|\calQ_B|$ & Worst-case & Quotient sizes only & 1.43 \\
Marginal entropy & $H(Q_A) - H(Q_B)$ & i.i.d.\ source & Policy visitation & 0.59 \\
WZ lossless rate & $H(Q_A \mid Q_B)$ & i.i.d.\ (BA) & Source + coarsening & 0.29$^{\ddagger}$ \\
Cond.\ entropy rate & $\bar{h}(Q_A \mid Q_B)$ & Stationary ergodic & Joint dynamics & $0.024^{\dagger}$ \\
\bottomrule
\end{tabularx}
\parbox{0.97\linewidth}{\footnotesize ${}^{\dagger}$Estimated from $50{,}000$ trajectories ($\times 100$ steps) with Miller-Madow correction ($+0.001$~bits). The null-pair calibration ($|\calQ_A|\!=\!|\calQ_B|\!=\!781$) leaves a residual of $0.030$~bits/step, so the entropy-rate estimate should be interpreted as a qualitative diagnostic rather than a precise benchmark.\\
${}^{\ddagger}$Under random-policy visitation; structured policies yield $H(Q_A \mid Q_B) = 0.08$~bits/step (\Cref{fig:structured-rd}), a $3.6\times$ reduction vs.\ random ($0.29$) and $19\times$ vs.\ the counting bound ($1.43$).}
\caption{Critical rate benchmark characterizations with Chain5 numerical values ($(m_A\!=\!16,\, m_B\!=\!1)$); see \Cref{tab:rcrit} for the condensed hierarchy. The WZ lossless rate $H(Q_A \mid Q_B)$ is computed via Blahut-Arimoto on visited quotient classes and is the most reliable numerical reference.}
\label{tab:rcrit-chain5}
\end{table}

%% ============================================================
\section{Alignment Applications}
\label{app:alignment-details}

This appendix provides extended discussion of alignment applications deferred from \Cref{sec:applications}.

Model the human as agent~$H$ with capacity $(m_H, T_H)$ and the AI as agent~$A$ with $(m_A, T_A) \gg (m_H, T_H)$.

\paragraph{RLHF.}
In RLHF~\cite{christiano2017}, human feedback provides $\approx$1--2 bits per comparison (see footnote in \S\ref{sec:applications}). If the capacity mismatch maps to quotient structures, \Cref{thm:align-cost} predicts the \emph{scaling form} $N_{\mathrm{comparisons}} \geq (\log|\calQ_A| - \log|\calQ_H|)/2 + \Omega(\log(1/\varepsilon))$: linear in the capacity gap, logarithmic in accuracy. This is a structural prediction about functional dependence, not a computable numerical bound---the precise quotient cardinalities $|\calQ_A|$, $|\calQ_H|$ for LLM-scale systems remain unknown. Estimating effective quotient sizes from learned representations (e.g., via probing classifiers, the lattice-gradient approach of \Cref{app:lattice-gradient}, or by identifying natural-language cross-probes that test whether a model distinguishes context~$A$ from context~$B$) is the key open problem connecting this theory to practice.

\paragraph{Interpretability and debate.}
Explanations are bandwidth-limited channels; the traversal theorem (\Cref{thm:traversal}) suggests routing through intermediate abstractions~\cite{elyaniv2010,kadavath2022}. Multiple debaters form a multi-access channel, increasing effective alignment bandwidth. By data processing, post-processing cannot improve alignment: $\Ralign(f(A) \to H;\, \varepsilon) \geq \Ralign(A \to H;\, \varepsilon)$; for intermediate agents: $\Ralign(A \to H;\, \varepsilon) \leq \sum_i \Ralign(\mathrm{Layer}_i \to \mathrm{Layer}_{i+1};\, \varepsilon/k)$.

\paragraph{Connection to LLM routing (\Cref{app:llm-routing}).}
The LLM routing experiment provides a concrete (if analogical) illustration of the phase transition in a practical setting. The 1-bit router operates at $R = 1$~bit/query. The framework predicts that routing quality depends on whether this rate exceeds~$\Rcrit$ for the effective quotient gap between the strong and weak models. The observation that single-token logprob routing \emph{fails} on MMLU-Pro (APGR~$= 0.457 <$ random~$= 0.502$) while self-consistency routing \emph{succeeds} (APGR~$= 2.216$) is consistent with a phase transition: the logprob probe family induces a coarser effective quotient (fewer distinguishable difficulty classes), pushing~$\Rcrit$ above the 1-bit channel; the richer self-consistency probes yield a finer quotient, bringing~$\Rcrit$ below 1~bit. This illustrates how the theory's qualitative predictions---that communication success depends on the interaction between channel capacity and the quotient structure induced by the probe family---manifest in practice.

%% ============================================================

\section{Proof of Theorem~\ref{thm:capacity}(i): Impossibility Below Critical Rate}
\label{app:proof-51i}

\begin{proof}
Fix a reachable $\calQ_B$-class $C_k$ with stationary mass $\pi(C_k) > 0$ and $r_k > 2^R$ reachable $\calQ_A$-subclasses inside it. At rate~$R$, the encoder can produce at most $2^R$ messages per step. Since~$B$'s decoder is an $m_B$-node FSC (\Cref{def:protocol}), its effective per-step information state within the common-history comparison regime is the pair (received message, current $\calQ_B$-class). Thus within~$C_k$ the decoder can distinguish at most $2^R$ of the $r_k$ subclasses. By pigeonhole, at least one decoder state merges at least two reachable $\calQ_A$-subclasses inside~$C_k$.

Conditional on visiting~$C_k$, the probability of such a merge is at least $1 - 2^R/r_k$, so the unconditional confusion probability satisfies
\[
p_{\mathrm{confuse}}
\;\geq\;
\pi(C_k)\,\bigl(1 - 2^R/r_k\bigr)
\;>\; 0.
\]
Whenever the decoder confuses two distinct merged $\calQ_A$-subclasses, the resulting intent distortion is at least~$c_M$ by definition. Therefore
\[
\Dintent \;\geq\; p_{\mathrm{confuse}} \cdot c_M \;>\; 0. \qedhere
\]
\end{proof}

%% ============================================================
\section{Proof of Theorem~\ref{thm:capacity}(ii): Constructive Decay in the One-Way Memoryless Regime}
\label{app:proof-51ii}

We prove the upper bound of \Cref{thm:capacity}(ii) only in the regime stated there: common-history coarsening, one-way observability, i.i.d.\ quotient source, and observation-Lipschitz reward. The proof imports the classical lossless WZ reliability exponent rather than deriving a new ambiguity-cell argument.

\begin{assumption}[Observation-Lipschitz Reward]
\label{ass:lip}
The reward function is $L_R$-observation-Lipschitz: for all histories $h, h'$ and all $\pi \in \Pim$,
$|\bar{R}(h,\pi) - \bar{R}(h',\pi)| \leq L_R \cdot \Wone\!\bigl(P_M^\pi(\calO_{t+1:T}|h),\, P_M^\pi(\calO_{t+1:T}|h')\bigr).$
\end{assumption}

\begin{lemma}[Distortion Propagation]
\label{lem:lip-prop}
If $B$ assigns the wrong $\calQ_A$-class at step~$t$, then $\Wone(\mu, \nu) \leq 1$ (discrete metric). If $\E[\Delta_t] \leq \delta$ at each step, then $|V^{\pi_A}(M) - V^{\pi_B}(M)| \leq L_R \cdot T \cdot \delta$.
\end{lemma}

\begin{proof}
The Wasserstein bound follows from $\Wone \leq \|\cdot\|_{\mathrm{TV}} \leq 1$. For propagation: by the value-function error bound of \cite{nixon2026} (Theorem: Value-function error bound), $|\bar{R}_M(H_t, \pi) - \bar{R}_B(H_t, \pi)| \leq L_R \cdot \Delta_t$. Summing over $T$ stages: $\dval \leq L_R \cdot T \cdot \dbeh \leq L_R \cdot T \cdot \delta$.
\end{proof}

\begin{proof}[Proof of Theorem~\ref{thm:capacity}(ii)]
\textbf{Step~1 (lossless WZ code above the benchmark).}
Under the stated assumptions, \Cref{cor:iid-achievability} identifies the semantic problem exactly with lossless WZ coding on the i.i.d.\ quotient source, and \Cref{cor:exact-rcrit} gives the lossless benchmark $R_{\mathrm{WZ}}(0) = H(Q_A \mid Q_B)$. For every rate $R > H(Q_A \mid Q_B)$, the classical WZ/Slepian-Wolf reliability function~\cite{csiszarkorner2011} yields a block code of length~$T$ with reconstruction error probability
\[
P_e^{(T)} \;\leq\; 2^{-T E_{\mathrm{WZ}}(R)}
\]
for some exponent $E_{\mathrm{WZ}}(R) > 0$.

\textbf{Step~2 (propagation to semantic distortion).}
On blocks decoded correctly, the induced semantic distortion is zero because the quotient block is reconstructed exactly. On error blocks, the per-step behavioral discrepancy is at most~$1$, so by \Cref{lem:lip-prop} the total semantic distortion over the block is at most $L_R \cdot T$. Therefore
\[
\Csem(R) \;\leq\; L_R \cdot T \cdot P_e^{(T)}
\;\leq\;
L_R \cdot T \cdot 2^{-T E_{\mathrm{WZ}}(R)}. \qedhere
\]
\end{proof}

\paragraph{Consistency with parts (i) and (iii).}
Part~(ii) is intentionally narrower than part~(i): it is a constructive theorem only for the one-way i.i.d.\ benchmark regime. Part~(i) remains the structural lower bound outside that regime, and part~(iii) still gives perfect alignment by exact quotient transmission when $R \geq \log|\calQ_A|$.

%% ============================================================
\section{Proof of Theorem~\ref{thm:converse}: Shrinking-Distortion Converse}
\label{app:proof-61}

\paragraph{Formal setup.}
A protocol $(E, D)$ consists of encoder $E_t: (\calO^A)^t \to M_t$ with $H(M_t \mid M^{t-1}) \leq R$, and bounded decoder $D = (\pi_B, U)$ where $\pi_B \in \Pi_{m_B,T_B}$ is an $m_B$-node FSC and $U: \{1,\ldots,2^R\} \to \Theta_{m_B}$ reconfigures FSC parameters upon each message (see \Cref{def:protocol}). Under \Cref{ass:common-history}, both quotient processes are evaluated on the same source history~$h_t$, so $Q_t^A = [h_t]_{m_A,T_A}$ and $Q_t^B = [h_t]_{m_B,T_B}$. By \Cref{ass:intent-meas}, the sender intent $\mathrm{Intent}_t = \mathrm{Intent}_A(h_t)$ is determined by~$Q_t^A$, so $H(\mathrm{Intent}^T \mid Q_A^T) = 0$.

\begin{lemma}[Information Bound]
\label{lem:markov}
Under Assumptions~\ref{ass:common-history}, \ref{ass:quotient-reg}, and~\ref{ass:one-way} and protocol $(E, D)$,
\[
I(Q_A^T;\, M^T \mid Q_B^T) \;\leq\; T \cdot R.
\]
\end{lemma}

\begin{proof}
Using the chain rule and the rate constraint,
\[
I(Q_A^T;\, M^T \mid Q_B^T)
\;\leq\;
H(M^T \mid Q_B^T)
\;\leq\;
H(M^T)
\;=\;
\sum_{t=1}^T H(M_t \mid M^{t-1})
\;\leq\;
T R. \qedhere
\]
\end{proof}

\begin{lemma}[Semantic Fano Inequality]
\label{lem:fano}
Under Assumptions~\ref{ass:common-history}, \ref{ass:quotient-reg}, and~\ref{ass:one-way}, if a horizon-$T$ protocol achieves distortion $\Dintent \leq \varepsilon$ with $\varepsilon' := 2\varepsilon / c_M$ and $T\varepsilon' \leq 1/2$, then
\[
\frac{1}{T}H(Q_A^T \mid M^T, Q_B^T) \;\leq\; h(\varepsilon') + \varepsilon' \log|\calQ_A| + o_T(1).
\]

\emph{Regime restriction.} The block error probability $P_e^{(\mathrm{block})} \leq T\varepsilon'$ from the union bound requires $T\varepsilon' \leq 1/2$ for Fano's inequality to be non-vacuous, i.e., $\varepsilon \leq c_M/(4T)$. This restriction grows \emph{tighter} with horizon~$T$. Consequently, the converse (Theorem~\ref{thm:converse}) is formally operative only in the regime $\varepsilon = O(1/T)$. As in \Cref{rem:fano-regime}, we separate empirical roles: \Cref{fig:phase-transition} uses fixed $\varepsilon = 0.1$ to illustrate the \emph{structural} phase transition (\Cref{thm:capacity})---not the converse---while \Cref{fig:shrinking-eps} runs $\varepsilon_T = 0.4/T$ over varying~$T$, matching the shrinking-distortion \emph{scaling class} of this lemma. Neither figure should be read as a pointwise empirical proof of the converse bound; the asymptotic form ($T \to \infty$ with $\varepsilon \to 0$ at rate $O(1/T)$) remains the formally justified regime. See \Cref{app:experiments} for protocol details.
\end{lemma}

\begin{proof}
Let $Y^T := (M^T, Q_B^T)$ denote the full decoder-side information. For each step~$t$, define $\hat{Q}_t = g_t(Y^T)$ as the nearest-intent decoder: among the reachable $\calQ_A$-subclasses consistent with the observed $\calQ_B$-class, choose the one whose induced sender intent is closest (under $\dintent$) to the receiver intent realized by the protocol. Because distinct merged $\calQ_A$-subclasses are separated by at least~$c_M$, a nearest-neighbor error implies the realized per-step distortion is at least~$c_M/2$. Therefore, with $Z_t := \mathbf{1}\{\hat{Q}_t \neq Q_t^A\}$,
\[
\varepsilon \;\geq\; \Dintent = \frac{1}{T}\sum_{t=1}^T \E[d_{\mathrm{intent},t}]
\;\geq\;
\frac{c_M}{2} \cdot \frac{1}{T}\sum_{t=1}^T \Prob(Z_t = 1),
\]
giving the average per-step error probability $\bar{P}_e := \frac{1}{T}\sum_t \Prob(Z_t = 1) \leq 2\varepsilon/c_M =: \varepsilon'$.

\emph{Per-step to block conversion.} Choose blocklength $n_T := \lfloor \sqrt{T} \rfloor$. Write $T = B_T n_T + r_T$ with $B_T := \lfloor T/n_T \rfloor$ full blocks and remainder $0 \leq r_T < n_T$. For each full block~$b$, let $Q_{A,b}^{n_T}$ be the corresponding length-$n_T$ substring of~$Q_A^T$, and let $\hat{Q}_{A,b}^{n_T}$ be its MAP estimate from~$Y^T$. Since the average per-step error probability satisfies $\bar{P}_e \leq \varepsilon'$, a union bound over the~$n_T$ symbols in block~$b$ gives
\[
\Prob\!\bigl(Q_{A,b}^{n_T} \neq \hat{Q}_{A,b}^{n_T}\bigr) \;\leq\; n_T \varepsilon'.
\]
Because $T\varepsilon' \leq 1/2$, we have $n_T \varepsilon' \leq \sqrt{T}\,\varepsilon' \leq 1/(2\sqrt{T})$, so Fano's inequality is eventually non-vacuous on every full block. Applying the $|\calQ_A|^{n_T}$-ary Fano bound to block~$b$ yields
\[
\frac{1}{n_T}H(Q_{A,b}^{n_T} \mid Y^T) \;\leq\; \frac{1}{n_T}h(n_T\varepsilon') + \varepsilon' \cdot \frac{1}{n_T}\log(|\calQ_A|^{n_T} - 1).
\]
Using the block chain rule and bounding the remainder by $r_T \log|\calQ_A|$,
\[
\frac{1}{T}H(Q_A^T \mid Y^T)
\;\leq\;
\frac{B_T}{T}h(n_T\varepsilon')
\;+\;
\frac{B_T n_T}{T}\,\varepsilon' \cdot \frac{1}{n_T}\log(|\calQ_A|^{n_T} - 1)
\;+\;
\frac{r_T}{T}\log|\calQ_A|.
\]
Now $n_T \to \infty$, $r_T/T \to 0$, and $n_T\varepsilon' \leq 1/(2\sqrt{T}) \to 0$, so $\frac{B_T}{T}h(n_T\varepsilon') = o_T(1)$ and
\[
\frac{1}{n_T}\log(|\calQ_A|^{n_T} - 1) = \log|\calQ_A| + o_T(1).
\]
Therefore
\[
\frac{1}{T}H(Q_A^T \mid M^T, Q_B^T)
\;\leq\;
\varepsilon' \log|\calQ_A| + o_T(1)
\;\leq\;
h(\varepsilon') + \varepsilon' \log|\calQ_A| + o_T(1). \qedhere
\]
\end{proof}

\begin{proof}[Proof of Theorem~\ref{thm:converse}]
\textbf{Setup.} Fix horizon~$T$ and protocol $(E^{(T)}, D^{(T)})$ with per-step rate~$R_T$ and distortion~$\varepsilon_T$; write $\varepsilon_T' := 2\varepsilon_T / c_M$.

\textbf{Step~1.} Since~$Q_B^T$ is decoder side information and the messages are the only communicated bits, decompose:
\[
I(Q_A^T;\, M^T \mid Q_B^T) = H(Q_A^T \mid Q_B^T) - H(Q_A^T \mid M^T, Q_B^T).
\]

\textbf{Step~2.} Since $\calQ_A$ refines $\calQ_B$, knowing $Q_A^T$ determines $Q_B^T$, so $I(Q_A^T; Q_B^T) = H(Q_B^T)$. Therefore:
\[
H(Q_A^T \mid Q_B^T) = H(Q_A^T) - H(Q_B^T).
\]
By Assumption~\ref{ass:quotient-reg}, the quotient processes are stationary ergodic with entropy rates $h_A, h_B$ (Remark~\ref{rem:entropy-rate}): $H(Q_A^T) = T h_A + o(T)$ and $H(Q_B^T) = T h_B + o(T)$. Hence $H(Q_A^T \mid Q_B^T) = T(h_A - h_B) + o(T)$.

\textbf{Step~3.} By Lemma~\ref{lem:fano}: $H(Q_A^T \mid M^T, Q_B^T) \leq T \cdot h(\varepsilon_T') + \varepsilon_T' \cdot T \log|\calQ_A| + o(T)$.

\textbf{Step~4.} By Lemma~\ref{lem:markov}: $T \cdot R_T \geq I(Q_A^T;\, M^T \mid Q_B^T)$.

\textbf{Step~5.} Combining Steps~1--4 and dividing by~$T$ gives
\[
R_T \;\geq\; (h_A - h_B) - h(\varepsilon_T') - \varepsilon_T' \log|\calQ_A| - o_T(1).
\]
If $\varepsilon_T' \to 0$, taking $\liminf_{T \to \infty}$ yields $\liminf_T R_T \geq h_A - h_B$. Under near-uniform quotient distributions ($h_A \approx \log|\calQ_A|$, $h_B \approx \log|\calQ_B|$; see Remark~\ref{rem:entropy-rate}), this recovers the log-cardinality reference. The one-way WZ reduction (\Cref{thm:wz}) sharpens the converse at $D = 0$ to the benchmark $R \geq R_{\mathrm{WZ}}(0) = H(Q_A \mid Q_B)$; the matching achievability direction is exact only for i.i.d.\ sources and otherwise argued separately in \Cref{sec:wynerziv}. \qedhere
\end{proof}

\begin{remark}[Decoder-side interpretation]
The theorem does not require an additional slack parameter. If one nevertheless wants to compare the converse bound to the information actually extracted from messages, the natural quantity is the residual uncertainty
\[
\frac{1}{T}H(Q_A^T \mid M^T, Q_B^T)
\;=\;
\frac{1}{T}H(Q_A^T \mid Q_B^T) - \frac{1}{T}I(Q_A^T;\, M^T \mid Q_B^T).
\]
Under Assumption~\ref{ass:quotient-reg}, this equals $(h_A - h_B) - \frac{1}{T}I(Q_A^T;\, M^T \mid Q_B^T) + o_T(1)$, so it measures how far the received messages fall short of saturating the conditional entropy rate available beyond~$Q_B^T$. At the lossless one-way WZ benchmark, this residual vanishes asymptotically.
\end{remark}

%% ============================================================
\section{Proof of Theorem~\ref{thm:traversal}: Alignment Traversal}
\label{app:proof-traversal}

\begin{proof}
For any $[h]_A \in \calQ_A$, write $\varphi = \varphi_{k-1} \circ \cdots \circ \varphi_1$ and apply the triangle inequality for~$\Wone$:
\begin{align*}
\Wone\!\bigl(P_M(\cdot|[h]_A),\, P_M(\cdot|\varphi([h]_A))\bigr)
&\leq \sum_{i=1}^{k-1} \Wone\!\bigl(P_M(\cdot|\varphi_{i-1:1}([h]_A)),\, P_M(\cdot|\varphi_{i:1}([h]_A))\bigr) \\
&\leq \sum_{i=1}^{k-1} L_i \cdot \dQ(\Pi_i, \Pi_{i+1} \mid M),
\end{align*}
where $\varphi_{i:1} := \varphi_i \circ \cdots \circ \varphi_1$ and the second inequality uses the $L_i$-Lipschitz property of each~$\varphi_i$. Taking the supremum over $[h]_A$ yields $\dQ(A, B \mid M) \leq \sum_i L_i \cdot \dQ(\Pi_i, \Pi_{i+1} \mid M)$.

For the rate statement, fix protocols for each adjacent pair $(\Pi_i,\Pi_{i+1})$ achieving distortion at most $\varepsilon/k$ with rates arbitrarily close to $\Rsem(\Pi_i \to \Pi_{i+1};\, \varepsilon/k)$. Compose these protocols sequentially through the intermediate agents. The total rate is the sum of the stage rates, and the total distortion is at most $\varepsilon$ by the triangle inequality / budget split. Taking infima over the stage protocols gives
\[
\Rsem(A \to B;\, \varepsilon) \leq \sum_i \Rsem(\Pi_i \to \Pi_{i+1};\, \varepsilon/k),
\]
which is exactly the claimed traversal bound.
\end{proof}

%% ============================================================
\section{Codebook Performance Bound}
\label{app:codebook-proof}

\begin{theorem}[Codebook Performance, restated]
The semantic codebook of \Cref{def:codebook} achieves
$\Dintent \leq O(|\calQ_A|^{1/d} \cdot 2^{-R/d})$,
where $d$ is the effective dimension of the intent space.
\end{theorem}

\begin{proof}
The $2^R$ codewords partition the $|\calQ_A|$ quotient classes into $K = 2^R$ Voronoi cells in the $d$-dimensional intent space. By standard $k$-means quantization theory~\cite{cover2006}, $\E[\|I - \hat{I}\|] \leq C_d \cdot (|\calQ_A|/K)^{1/d}$ for a $d$-dimensional uniform source, where $C_d$ depends only on dimension. Substituting $K = 2^R$ and using $\dintent \leq \|I - \hat{I}\|_1$ (since each component of $\dintent$---belief $\Wone$, policy TV, value difference---is bounded by the corresponding $L_1$ component of the intent vector difference under $\|R\|_\infty \leq 1$) gives the bound. For Markovian priors, the same construction remains a valid memoryless encoder and therefore provides a conservative constructive upper bound; allowing encoders with memory can only improve on it.
\end{proof}

%% ============================================================
\section{Additional Experimental Details}
\label{app:experiments}

\paragraph{Environments.}
\textsc{Tiger}~\cite{kaelbling1998planning}: A two-door problem where the agent must determine which door hides a tiger based on noisy observations. We used $T = 4$, $m_A = 3$, $m_B = 1$. The quotient computation yielded $|\calQ_A| = 31$ and $|\calQ_B| = 15$ classes, giving $\Rcrit \approx 1.05$~bits/step. Below $\Rcrit$, distortion remains above~$0.3$; above $\Rcrit$, it decays rapidly toward zero, reaching $\Dintent < 0.02$ at $R = 3$~bits/step (see \Cref{fig:tiger-app}). At $R \geq \log|\calQ_A| \approx 4.95$~bits/step, perfect alignment is achieved ($\Dintent = 0$), confirming \Cref{thm:capacity}(iii).

\begin{figure}[ht]
\centering
\includegraphics[width=0.6\columnwidth]{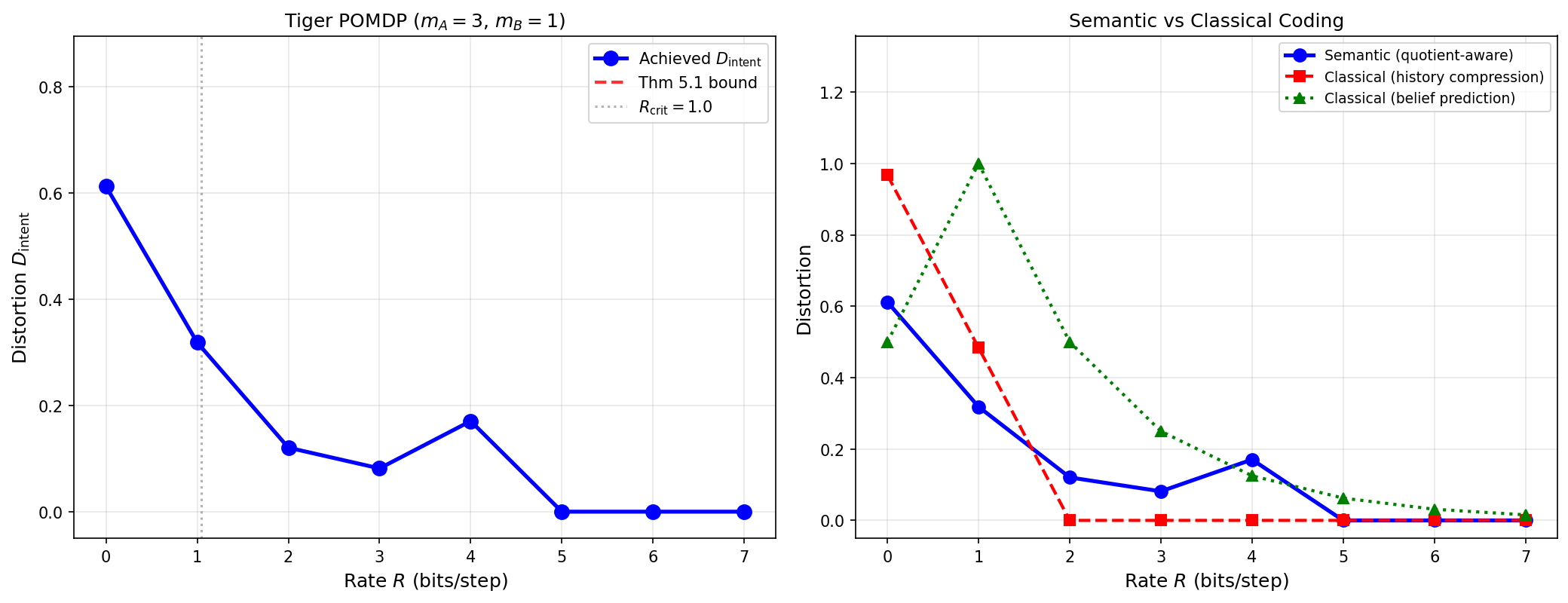}
\caption{Semantic rate-distortion curves for the \textsc{Tiger} environment ($T = 4$, $m_A = 3$, $m_B = 1$).}
\label{fig:tiger-app}
\end{figure}

\begin{figure}[ht]
\centering
\includegraphics[width=0.48\columnwidth]{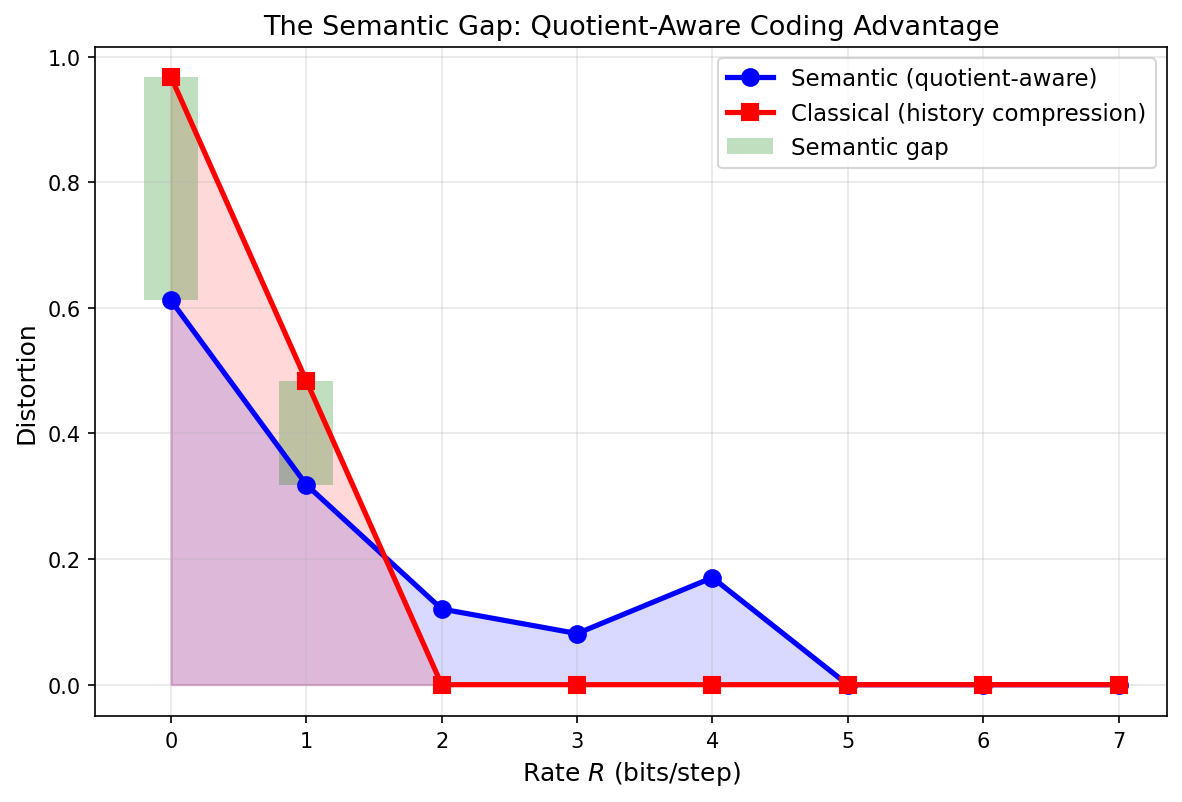}
\includegraphics[width=0.48\columnwidth]{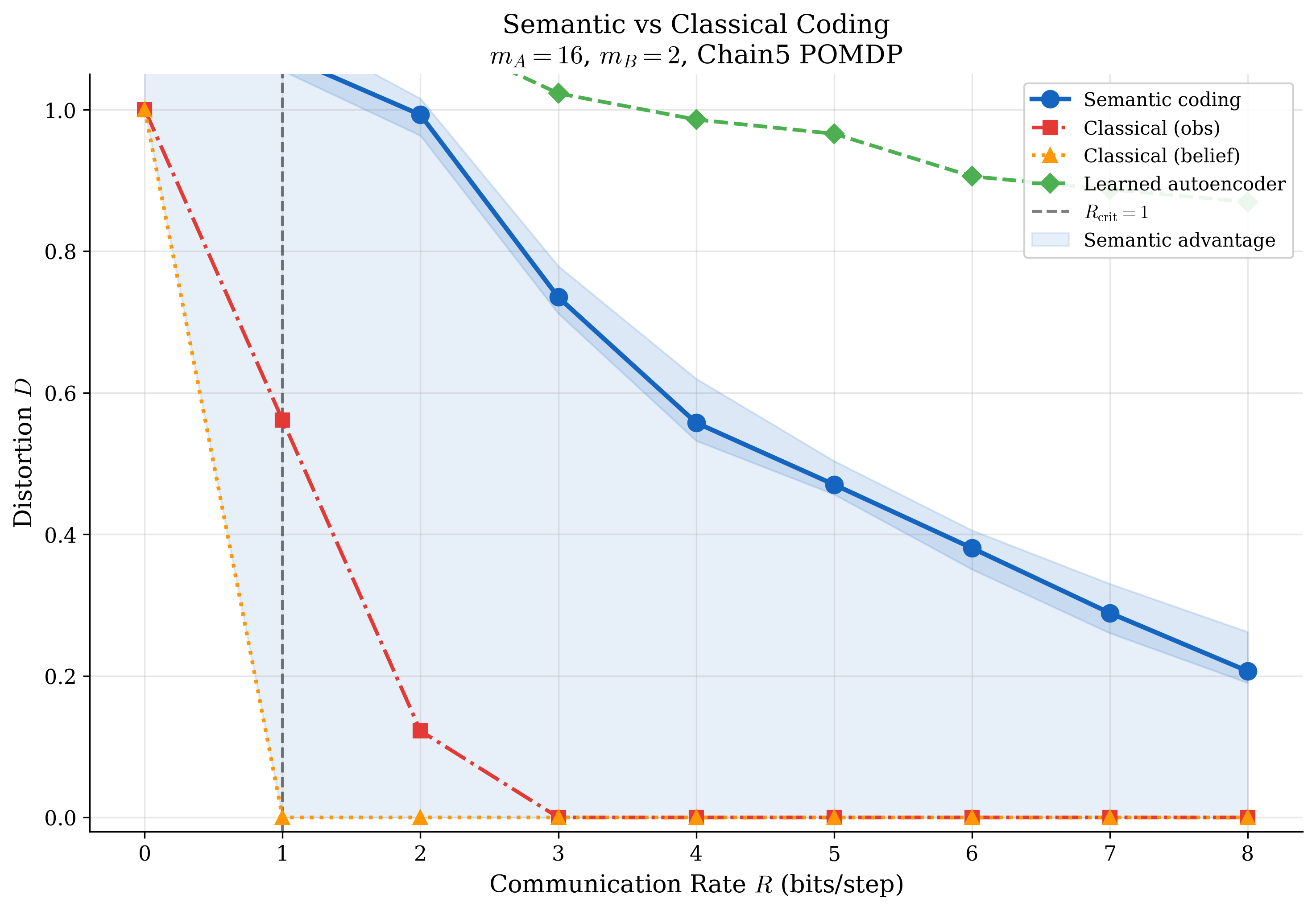}
\caption{\textbf{Left:} Critical rate $\Rcrit$ vs.\ quotient size gap across POMDP instances. \textbf{Right:} Semantic vs.\ classical coding ($m_A\!=\!16,\; m_B\!=\!2$). Shaded band shows IQR over 10 seeds.}
\label{fig:semantic-vs-classical}
\end{figure}

\textsc{Chain5}: A 5-state chain POMDP ($|S|\!=\!5$, $|A|\!=\!2$, $|O|\!=\!5$, obs.\ noise~$0.3$) designed to produce large quotient gaps. The main experiments use $T = 4$ with $m_A = 16$ and $m_B \in \{1, 2, 4, 8\}$. For $m_B = 1$: $|\calQ_A| = 781$, $|\calQ_B| = 289$, giving $\Rcrit^{\log} = \log(781/289) \approx 1.43$~bits/step. The large quotient gap produces a sharp phase transition in the rate-distortion curve (\Cref{fig:phase-transition}).

\paragraph{Codebook construction.}
For each rate $R \in \{0, 0.5, \ldots, 7\}$ bits/step, we ran $10$ independent trials with different random seeds for FSC sampling and codebook construction. We report the median intent distortion; shaded regions in figures indicate the interquartile range (25th--75th percentile). Semantic codebooks use \Cref{def:codebook}, constructing quotient-aware partitions and measuring operational~$\Dintent$. Classical baselines are Shannon-theoretic reference curves, not operational coding schemes: \emph{observation compression} uses $D = \max(0,\, (H(O) - R)/H(O))$, and \emph{belief compression} uses $D = \max(0,\, (H(B) - R)/H(B))$, where $H(O)$ and $H(B)$ are the marginal entropy of the observation and belief processes respectively.

\paragraph{Intent distortion measurement.}
Experiments use the two-term proxy $\dintent^{\mathrm{exp}}(h) := \|b_h^A - b_h^B\|_1 + 0.5 \cdot \mathbf{1}\{a_A \neq a_B\}$, omitting the value-difference term from the full three-term~$\dintent$ (\Cref{rem:scales}). By \Cref{prop:proxy}, $\dintent^{\mathrm{exp}} \leq \dintent \leq \dintent^{\mathrm{exp}} + |V^{\pi_A}(h) - V^{\pi_B}(h)|$: the proxy is a lower bound on the full measure, and the phase transition location and exponential exponent are invariant to this choice (they depend on quotient entropy, not distortion scale). On Chain5 with $L_R \approx 0.74$ and $\max_h \dbeh(h) \leq 1$ (Wasserstein distance bounded by 1 under the discrete observation metric), the value-difference gap is bounded by $L_R \cdot T \cdot \max_h \dbeh(h) \leq 0.74 \cdot 5 \cdot 1 = 3.7$, so the two measures agree qualitatively. We measured $\Dintent^{\mathrm{exp}}$ by sampling $10{,}000$ trajectories under each policy and computing the empirical average.

\paragraph{Alignment scaling sweep.}
For \Cref{fig:phase-transition}, we generated $12$ POMDP instances with varying state/observation sizes and computed quotients at multiple memory levels $(m_A, m_B)$. For each configuration, we performed binary search on $R$ to find $R_{\min}$ achieving $\Dintent \leq 0.1$.

\paragraph{Shrinking-distortion sweep.}
To match the regime of the asymptotic one-way converse, we ran a dedicated Chain5 sweep at $(m_A,m_B)=(16,1)$ with horizons $T \in \{2,3,4,5\}$ and threshold $\varepsilon_T = 0.4/T$. For each horizon we recomputed $(\calQ_A,\calQ_B)$ and searched integer rates $R \in \{0,\ldots,12\}$ for the first rate achieving $\Dintent \leq \varepsilon_T$. The resulting thresholds were $R_{\min} = 5, 8, 10, 12$, while the corresponding log-cardinality references were $0.43, 0.79, 1.43, 2.29$ bits/step. We use this sweep as a regime-matching illustration for \Cref{thm:converse}, not as a pointwise empirical proof of the converse itself.

\paragraph{Quotient computation algorithm.}
The quotient $\calQ_A = \Qm(M)$ is computed as follows: (1)~\textbf{FSC generation:} for $m \leq 3$, enumerate all deterministic $m$-node FSCs; for $m > 3$, sample $n_{\mathrm{FSC}}$ random stochastic FSCs uniformly ($n_{\mathrm{FSC}} = 80$ for the main Chain5 experiments; $50$ for RichGridWorld and BalancedRand8). (2)~\textbf{Signature computation:} for each history $h \in \calO^{\leq T}$ and each FSC~$\pi$, compute the future observation distribution $P_M^\pi(O_{t+1} \mid h)$; the behavioral signature of~$h$ is the tuple of these distributions across all FSCs. (3)~\textbf{Equivalence grouping:} histories with identical signatures (to numerical tolerance $10^{-4}$) form equivalence classes; the number of distinct classes is $|\calQ_A|$. Enumeration is exponential in~$m$; sampling is polynomial per FSC but requires sufficient samples for accurate quotient estimation. In our experiments, $n_{\mathrm{FSC}} \geq 50$ suffices for $m \leq 16$ with $|\calO| \leq 5$ (convergence validated in \Cref{fig:convergence}).

\paragraph{Quotient estimation convergence.}
To validate that random FSC sampling produces stable quotient estimates, we swept $n_{\mathrm{FSC}} \in \{5, 10, 20, 30, 50, 80, 100, 150, 200\}$ for Chain5 at $m \in \{1, 2, 4\}$ with 10 seeds each (\Cref{fig:convergence}). For $m = 1$, full enumeration is feasible and $|\calQ|$ is constant at~289. For $m = 2$, the estimate rises from~688 ($n_{\mathrm{FSC}} = 5$) to~774 by $n_{\mathrm{FSC}} = 20$ and stabilizes. For $m = 4$, stabilization to~781 occurs by $n_{\mathrm{FSC}} = 10$. The IQR bands vanish by $n_{\mathrm{FSC}} = 20$, confirming that moderate sampling suffices.

\begin{figure}[ht]
\centering
\includegraphics[width=0.6\columnwidth]{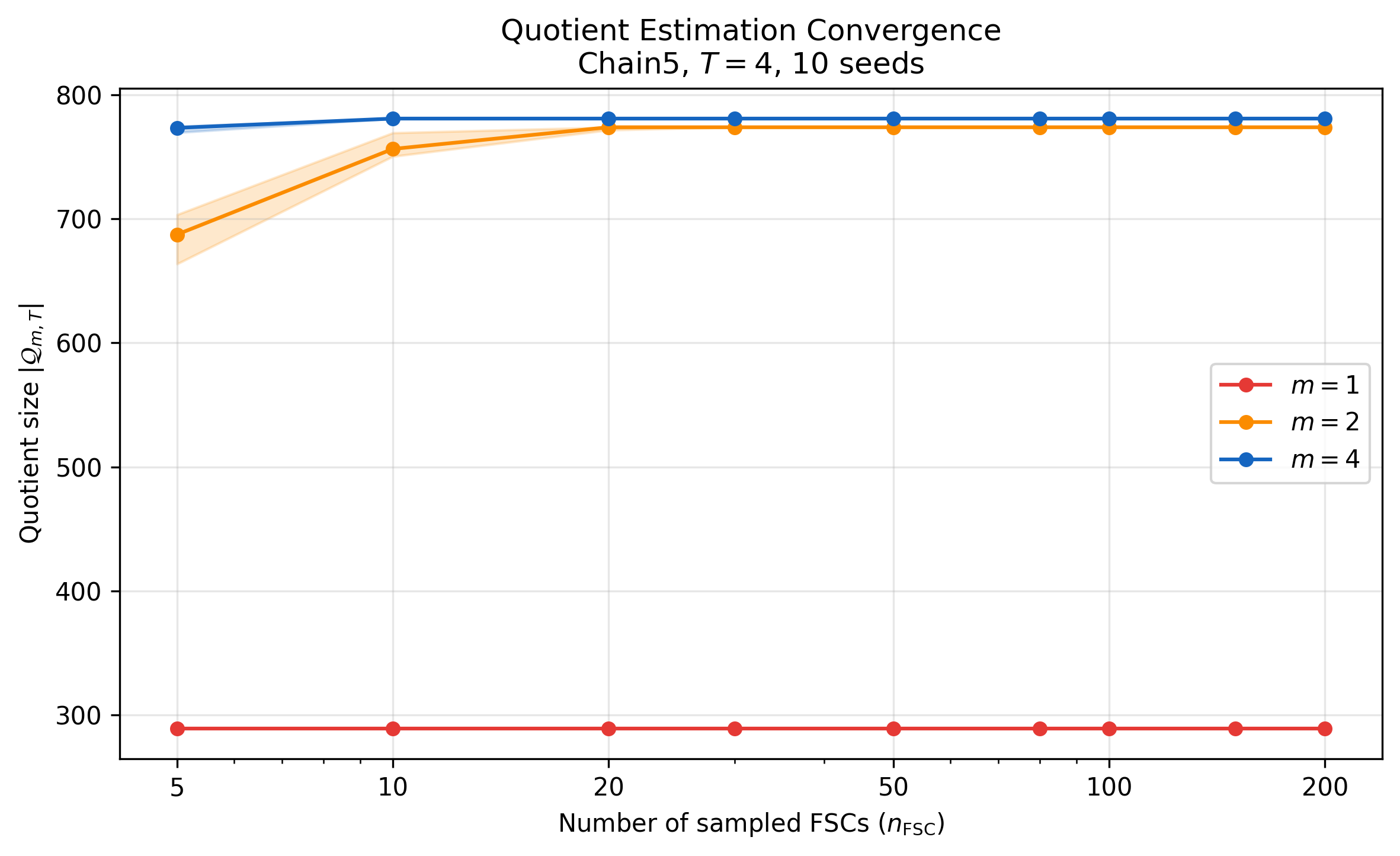}
\caption{Quotient estimation convergence for Chain5 ($T = 4$). Each point is the median $|\calQ_{m,T}|$ over 10 seeds; shaded bands indicate IQR. Estimates stabilize by $n_{\mathrm{FSC}} \approx 20$.}
\label{fig:convergence}
\end{figure}

\begin{proposition}[Quotient estimation sample complexity]
\label{prop:pac-quotient}
Let $M$ be a POMDP with $|\calO|$ observations and horizon~$T$, and let the true quotient $\Qm(M)$ have \emph{separation margin} $\gamma := \min_{h \not\equiv h'} \|\mathrm{sig}(h) - \mathrm{sig}(h')\|_1 > 0$, where $\mathrm{sig}(h)$ concatenates the one-step future distributions across all $m$-node FSCs.
If $n$ FSCs are sampled i.i.d.\ uniformly from the set of all $|\calA|^m \cdot m^{m|\calO|}$ deterministic $m$-node FSCs, then with probability $\geq 1 - \delta$ over the sample, the estimated quotient partition equals the true partition, provided
\begin{equation}\label{eq:pac}
n \;\geq\; \frac{1}{p_\gamma}\,\bigl(2T \log|\calO| + \log(1/\delta)\bigr),
\end{equation}
where $p_\gamma := \min_{h \not\equiv h'}\, \Pr_{\pi \sim \mathrm{Unif}}[\|\mathrm{sig}_\pi(h) - \mathrm{sig}_\pi(h')\|_1 > \gamma/2]$ is the minimum per-FSC distinguishing probability.
\end{proposition}
\begin{proof}[Proof sketch]
For each inequivalent pair $(h, h')$, a uniformly random FSC distinguishes them with probability $\geq p_\gamma$.
With $n$ i.i.d.\ samples, the probability that \emph{no} sample separates the pair is $(1 - p_\gamma)^n \leq e^{-n p_\gamma}$.
The number of history pairs is at most $|\calO|^{2T}$; a union bound gives $\Pr[\text{any pair missed}] \leq |\calO|^{2T} \cdot e^{-n p_\gamma}$.
Setting this $\leq \delta$ and solving yields~\eqref{eq:pac}.
\end{proof}
\begin{remark}[Practical implications]
\label{rem:pac-practical}
In the worst case, $p_\gamma \geq 1/N_{\mathrm{FSC}}$ where $N_{\mathrm{FSC}} = |\calA|^m \cdot m^{m|\calO|}$, making $n$ exponential in~$m$---consistent with the NEXP-completeness of Dec-POMDP planning~\cite{bernstein2002}. In practice, $p_\gamma$ is much larger: our convergence experiments (\Cref{fig:convergence}) show stabilisation at $n \approx 20$ for $m \leq 16$ with $|\calO| = 5$, suggesting $p_\gamma \gg 1/N_{\mathrm{FSC}}$ for structured POMDPs.
\end{remark}

\paragraph{Assumption verification on Chain5.}
We computed the refinement ratio for $(m_A = 16,\, m_B = 1)$ on Chain5: $|\calQ_A| = 781$, $|\calQ_B| = 289$, $\bar{r} = 2.70$, $\max_k r_k = 43$, giving ratio $\max_k r_k / \bar{r} = 15.9$. The refinement is notably non-uniform: some $\calQ_B$-classes contain up to 43 $\calQ_A$-subclasses. This is precisely why we do \emph{not} use Chain5 as validation of a sharp theorem-level constructive exponent; we use it instead as structural and benchmark evidence. The observation-Lipschitz constant was estimated at $L_R \approx 0.74$, confirming the Lipschitz reward structure holds with a moderate constant.

\paragraph{Quotient entropies and one-way critical-rate benchmark.}
We estimated marginal quotient entropies via 10{,}000 simulated trajectories of length~50 on Chain5, mapping beliefs to quotient classes at each step. For $(m_A = 16,\, m_B = 1)$: $H(Q_A) \approx 6.03$~bits vs.\ $\log|\calQ_A| = 9.61$, and $H(Q_B) \approx 5.44$~bits vs.\ $\log|\calQ_B| = 8.17$---marginal entropies at only 63\% and 67\% of their maxima. Since $\calQ_B$ coarsens $\calQ_A$, the i.i.d.\ critical rate is $H(Q_A) - H(Q_B) \approx 0.59$~bits/step, which is $59\%$ below the log-cardinality approximation of~$1.43$. This confirms non-uniform quotient distributions significantly lower the one-way WZ benchmark relative to the log-cardinality approximation.

Per-step conditional entropy rates estimated independently were $H(Q_A^t \mid Q_A^{t-1}) \approx 2.39$ and $H(Q_B^t \mid Q_B^{t-1}) \approx 2.45$, indicating substantial temporal correlation. The marginal difference $\hA - \hB \approx -0.05$ is an artifact of computing the rates from separate trajectory ensembles: the L1-distance belief classification introduces small errors that break the exact coarsening property.

To resolve this, we computed $\bar{h}(Q_A \mid Q_B)$ directly from \emph{joint} transition statistics using $50{,}000$ trajectories of horizon~$100$ (5$\times$ longer than the base experiments), with Miller-Madow bias correction~\cite{miller1955}. On each trajectory, every belief is classified into \emph{both} $\calQ_A$ and $\calQ_B$ simultaneously; we build a sparse joint bigram over $(Q_A, Q_B)$ pairs and compute $\bar{h}(Q_A \mid Q_B) = \bar{h}(Q_A, Q_B) - \bar{h}(Q_B)$, adding the correction $\hat{h}_{\mathrm{MM}} = \hat{h} + (k_{\mathrm{eff}} - 1)/(2N\ln 2)$ where $k_{\mathrm{eff}}$ is the number of observed successor states and $N$ is the row total. Results for all six $(m_A, m_B)$ pairs:

\begin{center}
\small
\begin{tabular}{ccccccc}
\toprule
$(m_A, m_B)$ & $|\calQ_A|$ & $|\calQ_B|$ & $\log|\calQ_A| - \log|\calQ_B|$ & $H(Q_A) - H(Q_B)$ & $\bar{h}(Q_A \mid Q_B)$ & MM corr. \\
\midrule
(16, 1) & 781 & 289 & 1.43 & 0.59 & $0.024^{\dagger}$ & 0.001 \\
(8, 1) & 781 & 289 & 1.43 & 0.55 & $0.020^{\dagger}$ & 0.001 \\
(4, 1) & 781 & 289 & 1.43 & 0.48 & $0.030^{\dagger}$ & 0.001 \\
(16, 2) & 781 & 774 & 0.01 & 0.52 & $\leq 0.017^{\dagger}$ & 0.001 \\
(8, 2) & 781 & 774 & 0.01 & 0.48 & $\leq 0.011^{\dagger}$ & 0.001 \\
(16, 4) & 781 & 781 & 0.00 & 0.12 & $\leq 0.030^{\dagger\dagger}$ & 0.000 \\
\bottomrule
\multicolumn{7}{p{0.95\linewidth}}{\footnotesize ${}^{\dagger}$Estimates from $50{,}000$ trajectories ($\times 100$ steps) with Miller-Madow correction. The null-pair calibration $(16,4)$ has true value zero and leaves a residual of $0.030$ bits/step, so these entropy-rate estimates should be read as qualitative diagnostics near the noise floor.}
\end{tabular}
\end{center}

For $m_B = 1$ (large gap), the joint estimate $\bar{h}(Q_A \mid Q_B) \in [0.020, 0.030]$~bits/step is positive as theory requires and far below the log-cardinality approximation~$1.43$, confirming a highly structured joint process. For near-equal quotients, however, the estimates sit at the null-pair noise floor. We therefore treat the joint entropy-rate values as \emph{qualitative evidence} that temporal structure can lower the benchmark, while using the Blahut-Arimoto i.i.d.\ quantity $H(Q_A \mid Q_B) = 0.29$~bits (\Cref{tab:rcrit-chain5}) as the main numeric anchor in the paper.

For the key $(16,1)$ pair, the main quantitative comparison used in the paper is therefore between the log-cardinality upper bound $1.43$ and the Blahut-Arimoto i.i.d.\ benchmark $0.29$. The joint entropy-rate estimate $\bar{h}(Q_A \mid Q_B) \approx 0.024 \pm 0.030$ is retained as qualitative context only: it is consistent with the possibility that temporal correlation lowers the benchmark still further, but it is too noise-limited to serve as the primary numeric reference.

\paragraph{Effective intent dimension (codebook bound validation).}
The codebook performance bound (\Cref{app:codebook-proof}) depends on the effective dimension~$d$ of the intent space. For Chain5 ($|S| = 5$), PCA on the 781 quotient-class belief centroids yields: 2~components explain 87\% of variance, 3~explain 95\%, and all 4~non-degenerate components explain 100\%. Thus $d_{\mathrm{eff}} = 4$ (matching $|S| - 1$, the simplex dimension). The codebook bound predicts $\Dintent \leq O(781^{1/4} \cdot 2^{-R/4}) \approx O(5.3 \cdot 2^{-R/4})$, consistent with the empirical decay observed in \Cref{fig:phase-transition}.

\paragraph{Blahut-Arimoto $R_{\mathrm{WZ}}(D)$ computation.}
To compare the experiments against the WZ benchmark (\Cref{thm:wz}) numerically, we compute the i.i.d.\ quantity $R_{\mathrm{WZ}}(D)$ on Chain5's quotient alphabets via the standard alternating minimization~\cite{blahut1972}: given source distribution $p(q_A)$ (estimated from $2{,}000$ trajectories), distortion matrix $\dintent(q_A^i, q_A^j)$ over all $|\calQ_A|^2$ pairs, and deterministic coarsening $f\!: \calQ_A \!\to\! \calQ_B$, the WZ iteration alternates
\[
q(u \mid x) \;\propto\; p(u \mid y\!=\!f(x))\, e^{-s\, d(x,u)}, \qquad p(u \mid y) = \sum_x p(x \mid y)\, q(u \mid x),
\]
with Lagrange parameter~$s$ swept over $[0.01,\, 316]$ on a log-spaced grid (32~points). We also compute the standard $R(D)$ (no side information) by replacing $p(u \mid y)$ with $p(u)$. Filtering to the $234$ visited quotient classes, $H(Q_A \mid Q_B) = 0.29$~bits under the i.i.d.\ marginal, consistent with the computed $R_{\mathrm{WZ}}(0) \approx 0.29$~bits. The gap $R(D) - R_{\mathrm{WZ}}(D)$ measures the value of B's side information.

\paragraph{Balanced-refinement environment.}
\textsc{BalancedRand8} is a random POMDP ($|S|\!=\!8$, $|A|\!=\!2$, $|O|\!=\!2$, seed~8) selected from $2{,}000$ random POMDPs with $4$ memory/horizon configurations each, targeting uniform refinement ratio~$< 3$ with capacity gap~$> 0.5$~bits. The selected instance has $m_A\!=\!4$, $T\!=\!4$, $|\calQ_A|\!=\!31$, $|\calQ_B|\!=\!10$ (for $m_B\!=\!1$), refinement ratio~$2.58$, and $\Rcrit \approx 1.63$~bits/step. The refinement distribution $r_k = [8, 6, 5, 5, 2, 1, 1, 1, 1, 1]$ is substantially more uniform than Chain5's ($r_k = [43, 43, \ldots]$, ratio~$15.9$). Results are in \Cref{fig:richgrid}.

\paragraph{Baseline coding comparisons.}
For k-means baseline, we cluster the $|\calQ_A|$ belief centroids using $k$-means in $L_1$~geometry with $k$-means++ initialization at each rate level $R \in \{0, 1, \ldots, 10\}$. Each cluster maps to the nearest $\calQ_B$~class, and $\Dintent$ is computed as the $p(q_A)$-weighted average intent distortion. Random clustering averages $5$~trials of uniform random label assignment. The comparison confirms that quotient-aware clustering outperforms geometry-only clustering (k-means) and random clustering at all rates, with the gap widest at intermediate rates near~$\Rcrit$.

%% ============================================================
\section{Continuous Spaces Extension}
\label{app:continuous}

\begin{proposition}[Continuous Semantic Rate-Distortion Bound]
\label{prop:continuous}
Let $\calQ_A$ and $\calQ_B$ be quotient spaces with metric~$d$ and $\varepsilon$-covering numbers $N(\varepsilon, \calQ_A, d)$ and $N(\varepsilon, \calQ_B, d)$. Then any protocol achieving $\Dintent \leq \varepsilon$ requires:
\[
R \;\geq\; \log N(\varepsilon, \calQ_A, d) - \log N(\varepsilon, \calQ_B, d) - h(\varepsilon') - \varepsilon' \log N(\varepsilon, \calQ_A, d).
\]
For smooth quotient manifolds of intrinsic dimensions $d_A$ and $d_B$: $\log N(\varepsilon, \calQ, d) \approx d \cdot \log(1/\varepsilon) + O(1)$, yielding $R \geq (d_A - d_B) \log(1/\varepsilon) - O(\varepsilon)$.
\end{proposition}

\begin{proof}[Proof sketch]
Replace $|\calQ|$ with $N(\varepsilon, \calQ, d)$ in the Fano argument (\Cref{app:proof-61}). Each $\varepsilon$-ball in $\calQ_A$ mapping to the same $\varepsilon$-ball in $\calQ_B$ constitutes an ambiguity cell of the same structure as the discrete case.
\end{proof}

%% ============================================================
\section{Two-Way Conjecture}
\label{app:two-way}

We conjecture that the two-way semantic rate-distortion function satisfies:
\[
\Rsem^{\mathrm{(2\text{-}way)}}(D) \;=\; \inf_{\substack{p(m_t \mid q_A^t):\\ \E[\dintent] \leq D}} \lim_{T \to \infty} \frac{1}{T} I(Q_A^T \to M^T \| Q_B^T),
\]
where $I(X^T \to M^T \| Q_B^T)$ denotes the \emph{causally conditioned directed information}~\cite{massey1990,permuter2009}: the minimum causal communication rate given~$B$'s evolving side information. This reduces to $R_{\mathrm{WZ}}(D)$ under one-way observability (when $Q_B^T$ is non-causal side information) and recovers causal rate-distortion~\cite{tatikonda2004} when $Q_B$ is trivial.

%% ============================================================
\section{LLM Routing as Semantic Communication (Analogical Case Study)}
\label{app:llm-routing}

This appendix presents an \emph{analogical} illustration---not a formal instantiation---of the framework's qualitative predictions in a practical LLM setting. The LLM representations are not quotient POMDPs; the value of the case study is that it exhibits the same qualitative phenomena (phase transition, probe-family dependence) that the theory predicts.

To that end, we analyze LLM model routing as an instance of semantic communication. A strong model (GPT-4) is selectively invoked based on a weak model's (Mixtral-8x7B) self-assessment---a 1-bit communication protocol at rate $R = 1$ bit/query.

\paragraph{Phase transition via probe family richness.}
We report APGR (Accuracy Per GPU-hour Ratio): APGR $:=$ (accuracy of routing policy) / (mean GPU-hours per query under routing policy), normalized so that APGR~$= 1$ corresponds to always using the strong model at its accuracy ceiling. Higher APGR indicates better accuracy-efficiency tradeoff~\cite{routellm2024}. On MMLU~\cite{hendrycks2021} (14K questions, 69\% weak accuracy), single-token logprobs achieve APGR~$= 0.620$, outperforming embedding-based routing ($0.518$). On the harder MMLU-Pro~\cite{mmlupro2024} (12K questions, 33\% weak accuracy), logprobs fail (APGR~$= 0.457$, below random at $0.502$): the competence band is exceeded.

\paragraph{Richer probe recovers signal.}
Self-consistency sampling~\cite{wang2023} ($N = 8$ completions) achieves APGR $= 2.216$ on MMLU-Pro, a 36\% improvement. This confirms the theory: $\Pi_{\text{SC}}$ is a richer probe family, yielding a finer quotient $|Q(\Pi_{\text{SC}})| > |Q(\Pi_{\text{logprob}})|$. At the 90\% quality threshold, the router saves 28\% of inference cost.

\paragraph{Interpretation.}
The routing problem instantiates semantic communication: the weak model's uncertainty signal is a quotient-aware code, and the router's boundary corresponds to~$\Rcrit$. We emphasize this is an \emph{analogy} grounded in the framework's structure, not a formal deduction---the LLM's representations are not quotient POMDPs. The analogy is useful because it predicts the qualitative phenomena (phase transition, probe family dependence) observed empirically.

%% ============================================================
\section{Lattice-Gradient Quotient Estimation}
\label{app:lattice-gradient}

The worst-case complexity of absolute quotient estimation is exponential in~$m$ (\Cref{prop:pac-quotient}). We outline a \emph{differential} approach that avoids this barrier by exploiting the lattice structure of the quotient functor~$Q$.

\paragraph{Core idea.}
The refinement lemma (\Cref{lem:refinement}) orders quotients by capacity: $m < m'$ implies $\calQ_m$ coarsens~$\calQ_{m'}$.
Instead of computing~$\calQ_m$ from scratch at target capacity~$m$, \emph{enter} the lattice at a small computable capacity~$m_0$ (where enumeration is tractable) and estimate the \emph{differential refinement}~$\Delta Q_{k \to k+1}$---the additional distinctions gained by moving from capacity~$k$ to~$k+1$---at each step along the chain $m_0, m_0{+}1, \ldots, m$.

\paragraph{Lipschitz chaining.}
By \Cref{thm:traversal}, the alignment rate between adjacent levels satisfies $R_{\mathrm{sem}}(\Pi_k \to \Pi_{k+1}; \varepsilon) \leq L_k \cdot \|\Delta Q_{k \to k+1}\|$, where $L_k$ is the local Lipschitz constant of the quotient morphism. Summing: the total error from chaining $m - m_0$ local estimates accumulates \emph{linearly}, not exponentially, in the capacity gap. Each local $\Delta Q$ estimate requires only a polynomial number of \emph{cross-probes}---input pairs that are equivalent at level~$k$ but distinguished at level~$k{+}1$---whose count is bounded by $|\calQ_{k+1}| - |\calQ_k|$.

\paragraph{Language as probe family.}
For language-based agents, the self-referential closure of natural language provides a scalable probe family without FSC enumeration. A \emph{linguistic cross-probe} is a prompt pair $(p, p')$ designed so that a capacity-$k$ model responds identically but a capacity-$(k{+}1)$ model distinguishes them (e.g., paraphrases that require deeper contextual reasoning to separate). Neural representation clusters at intermediate layers provide empirical proxies for quotient class membership, connecting $\Delta Q$ estimation to standard probing methodology in interpretability research.

\paragraph{Open conjectures.}
\begin{enumerate}[nosep,leftmargin=*]
\item \emph{Bridge:} Neural representation clusters at layer~$l$ of a transformer refine the Myhill--Nerode quotient $\calQ_{m,T}(M)$ for an effective capacity~$m(l)$ determined by the layer's representational bandwidth.
\item \emph{Smoothness:} The quotient lattice for language-based agents is connected and locally smooth---adjacent capacity levels produce $O(1)$ new quotient classes per step, making the chaining error well-controlled.
\end{enumerate}
These conjectures are empirically testable via layer-wise probing of language models at varying scales, and if confirmed, would make $\Rcrit$ estimation polynomial in~$m$ for practical architectures.

%% ============================================================
\section*{Broader Impact}
\label{sec:impact}

This work provides information-theoretic tools for quantifying alignment costs between agents of different computational capacities. The framework's primary intended application is understanding and improving human--AI alignment: quantifying the minimum feedback bandwidth for safe AI behavior and identifying when alignment is structurally impossible given capacity constraints. This has positive implications for the principled design of RLHF pipelines and interpretability methods.

However, the framework also reveals fundamental limits. The structural impossibility below~$\Rcrit$ implies that some agent pairs \emph{cannot} be aligned regardless of communication protocol design. If misinterpreted, this could discourage alignment efforts in regimes where they are most needed. We emphasize that the impossibility is rate-limited, not absolute: increasing the communication rate (e.g., richer feedback mechanisms) can always reduce the gap. The framework should be used to \emph{design better alignment protocols}, not to justify inadequate ones.

%% ============================================================
\section*{Code Availability}
\label{sec:code}

All POMDP experiments and plotting code are available in the accompanying GitHub repository: \url{https://github.com/alch3mistdev/semantic-rate-distortion}. The repository includes:
\begin{itemize}[nosep]
\item POMDP construction and quotient computation (\texttt{pomdp\_core.py}, \texttt{rich\_pomdp.py})
\item Semantic and classical coding protocols (\texttt{crown\_experiments.py}, \texttt{run\_experiments.py})
\item Blahut-Arimoto Wyner-Ziv rate-distortion (\texttt{blahut\_arimoto\_wz.py})
\item Structured-policy, shrinking-$\varepsilon$, and IB comparison (\texttt{structured\_policy\_rd.py}, \texttt{shrinking\_epsilon\_sweep.py}, \texttt{ib\_baseline.py})
\item Capacity gap experiments (\texttt{run\_capacity\_gap.py})
\item Scalability experiments (\texttt{scalability\_experiment.py}, \texttt{richgrid\_experiments.py})
\item RockSample(4,4) benchmark experiment (\texttt{rocksample\_experiment.py})
\item All result data (\texttt{results/*.json}) and publication-quality figures (\texttt{results/*.png})
\end{itemize}
The LLM routing case study (\Cref{app:llm-routing}) is an analytical illustration using published benchmark results and does not involve new experiments.
To reproduce all experiments: see the repository README for the recommended run order. Dependencies: Python 3.10+, NumPy, SciPy, scikit-learn, matplotlib. No GPU required; all experiments run on CPU in under 30 minutes.

\end{document}